\documentclass[11pt]{article}
\usepackage[utf8]{inputenc}
\usepackage{geometry}
\usepackage{babel}
\usepackage{verbatim}
\usepackage{mathtools}
\usepackage{bm}
\usepackage{amsmath}
\usepackage{amssymb}
\usepackage{multicol}
\usepackage[ruled,noline,noend,linesnumbered]{algorithm2e}
\usepackage[dvipsnames,table,xcdraw]{xcolor}
\definecolor{linkColor}{HTML}{E74C3C}
\definecolor{pearcomp}{HTML}{B97E29}
\definecolor{citeColor}{HTML}{2980B9}
\definecolor{urlColor}{HTML}{1D2DEC}
\definecolor{conjColor}{HTML}{9ab569}
\usepackage[CJKbookmarks=true,
            bookmarksnumbered=true,
            bookmarksopen=true,
            colorlinks=true,
            citecolor=citeColor,
            linkcolor=linkColor,
            anchorcolor=red,
            urlcolor=urlColor,
            ]{hyperref}
\makeatletter

\usepackage{bm}
\usepackage{enumitem}
\usepackage{tikz}
\usetikzlibrary{calc,patterns,angles,quotes}

\tikzset{
    right angle quadrant/.code={
        \pgfmathsetmacro\quadranta{{1,1,-1,-1}[#1-1]}     
        \pgfmathsetmacro\quadrantb{{1,-1,-1,1}[#1-1]}},
    right angle quadrant=1, 
    right angle length/.code={\def\rightanglelength{#1}},   
    right angle length=2ex, 
    right angle symbol/.style n args={3}{
        insert path={
            let \p0 = ($(#1)!(#3)!(#2)$) in     
                let \p1 = ($(\p0)!\quadranta*\rightanglelength!(#3)$), 
                \p2 = ($(\p0)!\quadrantb*\rightanglelength!(#2)$) in 
                let \p3 = ($(\p1)+(\p2)-(\p0)$) in  
            (\p1) -- (\p3) -- (\p2)
        }
    }
}

\usepackage{babel}
\usepackage{subcaption}

\usepackage[authoryear,sort]{natbib}

\usepackage{algorithmic}

\usepackage{amsthm}
\usepackage{bbm}
\usepackage{times}
\usepackage{lipsum}

\theoremstyle{plain}

\newtheorem{lemma}{\textbf{Lemma}}
\newtheorem{theorem}{\textbf{Theorem}}

\newtheorem{assumption}{\textbf{Assumption}}

\newtheorem{definition}{\textbf{Definition}}

\newtheorem{procedure_con}{\textbf{Procedure}}
\theoremstyle{definition}
\newtheorem{remark}{\textbf{Remark}}

\DeclareMathOperator*{\argmax}{argmax}
\DeclareMathOperator*{\argmin}{argmin}
\newcommand\blfootnote[1]{%
  \begingroup
  \renewcommand\thefootnote{}\footnote{#1}%
  \addtocounter{footnote}{-1}%
  \endgroup
}
\usepackage[authoryear,sort]{natbib} 

\title{The Sample Complexity of Online Contract Design}
\author{
Banghua Zhu$^\dagger$ \quad
Stephen Bates$^{\dagger, \ddagger}$ \quad
Zhuoran Yang$^{\diamond}$ \quad
Yixin Wang$^{\mathsection}$ \quad \\
Jiantao Jiao$^{\dagger, \ddagger}$ \quad
Michael I. Jordan$^{\dagger, \ddagger}$ \blfootnote{$^\dagger$ Department of Electrical Engineering and Computer Sciences, UC Berkeley \quad 
$^\ddagger$ Department of Statistics, UC Berkeley \quad 
$\diamond$ Department of Statistics and Data Science, Yale University \quad
${\mathsection}$ Department of Statistics, University of Michigan}}
\begin{document}

\maketitle

\begin{abstract}
We study the hidden-action principal-agent contract design problem in an online setting. 
In each round, 
the principal posts a contract that specifies the payment to the agent based on each outcome. 
The agent then makes a strategic choice of action that maximizes her own utility, but the action is not directly observable by the principal. 
The principal observes the outcome and receives utility from the agent's choice of action. Based on past observations, the principal dynamically adjusts the  contracts with the goal of  maximizing her utility. 
We treat this problem as a continuum-armed bandit problem, where we think of each potential contract as an arm.

We introduce an online learning algorithm and provide an upper bound on its Stackelberg regret. 
In particular, we show that when the contract space is $[0,1]^m$, the Stackelberg regret is upper bounded by $\widetilde O(\sqrt{m} \cdot T^{1-1/(2m+1)})$, and lower bounded by $\Omega(T^{1-1/(m+2)})$, where $\widetilde O$ omits logarithmic factors. This result shows that exponential-in-$m$ samples are both sufficient and necessary to learn a near-optimal contract, resolving an open problem in~\citet{ho2016adaptive} on the hardness of online contract design.
Moreover, when contracts are restricted to some subset $\mathcal{F} \subset [0,1]^m$, we define an intrinsic dimension of $\mathcal{F}$ that depends on the covering number of the spherical code in the space and bound the regret in terms of this intrinsic dimension.
When $\mathcal{F}$ is the family of linear contracts, we show that the Stackelberg regret grows exactly as $\Theta(T^{2/3})$.

Technically, the contract design problem is challenging because the utility function is discontinuous. Indeed, bounding the discretization error in this setting has been an open problem~\citep{ho2016adaptive}. In this paper, we identify a limited set of directions in which the utility function is continuous, allowing us to design a new discretization method and bound its error. This approach enables the first upper bound with no restrictions on the contract and action space.  The techniques we introduce may extend more generally to other Stackelberg games and continuum-armed bandits.
\end{abstract}

\section{Introduction}

Contract theory studies the interactions between a principal and an agent when the two parties transact in the presence of private information~\citep{faure2001transaction, bolton2004contract, salanie2005economics}. 
The principal would like to achieve her desired outcomes by hiring agents to work for her. The agent wishes to make money by working for the principal. They develop agreements in the form of a contract, which specifies how much the principal would pay under the different possible outcomes of the agent's work.  

Recent years have brought new insights regarding the computational and representational aspects of the contract design problem \citep{babaioff2006combinatorial,  dutting2021complexity, guruganesh2021contracts, castiglioni2022designing, dutting2019simple, dutting2022combinatorial, carroll20152014, alon2021contracts, dutting2023ambiguous}. In this paper, we continue this line of work, studying the sample complexity of the \emph{online contract design problem}, where the principal repeatedly interacts with the agents and dynamically adjusts her proposed contract based on historical observations of the outcomes. The setting of online contract design is common in 
crowdsourcing markets, such as Amazon Mechanical Turk
and Microsoft’s Universal Human Relevance System. These crowdsourcing markets are platforms designed to match available
human resources with tasks to complete. Using these platforms, the principal may post tasks that they
would like completed, along with the amount of money they are willing to pay. Workers (agents) then choose
whether or not to accept the available tasks and complete the work. Their interactions may happen in multiple rounds so the principal can adjust the contract based on historical results. The online contract design problem  is also closely related to the insurance problem, sponsored content creation, affiliate marketing, and freelancing.

Formally, the contract design problem can be formulated as a hidden-action principal-agent problem \citep{grossman1992analysis}. Consider the simple case of  finite outcomes and finite actions. The agent can select one of the actions from the action set $\mathcal{A}$. 
Each action $a$ 
is associated with a distribution $p(\cdot | a)$ over
$m$ outcomes $\{o_1,\cdots,o_m\}$, and incurs a cost $c(a)\in \mathbb{R}_{\geq 0}$. The
principal has different values for the outcomes, denoted as $\vec v\in[0, 1]^m$. Based on the valuation, the principal designs a contract $\vec f\in\mathcal{F}\subset [0, 1]^m$ that specifies a payment
$f_i\in[0, 1]$ for each outcome $o_i$. 
The agent chooses an action   that maximizes expected payment minus cost, i.e., $a^\star(\vec f) = \argmax_{a\in\mathcal{A}}\mathbb{E}_{o\sim p(\cdot | a)}[f_o] - c(a)$. The principal seeks to set up a contract that maximizes her expected utility, which is defined as the expected value of the received outcome, subtracting the payment. 
That is, the utility of contract is given by $u(\vec f) = \mathbb{E}_{o\sim p(\cdot | a^\star(\vec f))}[v_o - f_o]$.


We consider a repeated contract design problem with bandit feedback, where in each round $t \in [T] $ the principal proposes a new contract $\vec f^{(t)}$ based on historical observations, and observes an  outcome $o^{(t)}$ sampled from the agent's best response distribution $p(\cdot|a^\star(f^{(t)}))$. The agent in each round may have a different private type in the sense that the distribution $p$ and the cost $c$ may both depend on a random variable $i$ that is unobserved by the principal. 
Such a private type leads to different best-response behaviors by the agents. 
We are interested in designing a good history-dependent policy $\pi$ that maps the history, $\mathcal{H}^{(t-1)}=(\vec f^{(1)}, o^{(1)}, \cdots, \vec f^{(t-1)}, o^{(t-1)})$, to a new (possibly random) contract $\vec f^{(t)}\in\mathcal{F}$. The target is to minimize the Stackelberg regret that compares the utility achieved by the optimal single deterministic contract with the utility from the policy $\pi$, defined as
\begin{align*}
    R_T(\pi,\mathcal{F}) = \sup_{\vec f\in\mathcal{F}}\sum_{t=1}^T \mathbb{E}_{\vec f^{(t)}\sim \pi(\cdot|\mathcal{H}^{(t-1)})}[u(\vec f) - u(\vec f^{(t)})].
\end{align*}

The online contract design problem can be viewed as a special case of repeated Stackelberg game~\citep{von2010market, marecki2012playing}, where the principal moves first and the agent follows. The main difficulty of the contract design problem is \emph{information asymmetry} --- the principal does not know the private types nor the choice of actions of the agents; such information is possessed only by the agents.

One possible perspective on such information asymmetry is to formulate the problem  as a black-box continuum-armed bandit problem, where each contract is an arm, and the utility of the contract for the principal is the reward function. To deal with the continuous arm space, one first discretizes the space $\mathcal{F}$ uniformly or adaptively, 
with the hope that the discretization error is small \citep{kleinberg2008multi, kleinberg2013bandits}. 
Subsequently, the problem is reduced to a finite-armed bandit problem, where standard bandit algorithms such as  Upper Confidence Bound (UCB) or successive elimination are known to provide a near-optimal solution~\citep{lattimore2020bandit}.

Unfortunately, simple uniform  discretization does not work for the contract design problem. Although the utility function $u(\vec f)$ is a piecewise linear function of $\vec f$, it is not continuous because the best-response mapping from the agent is not continuous in $\vec f$. Thus the discretization error can be arbitrarily large  for the naive uniform discretization.  
To address this challenge, we design a novel  discretization scheme for the continuous contract space. Our discretization uses coding theory, specifically we exploit the maximum packing of a spherical code~\citep{Manin_2019, kabatiansky1978bounds}.
We thereby bound the discretization error without any extra assumptions. This leads to a  near-optimal regret bound for both general contracts and linear contracts, resolving several open problems from prior work.

\subsection{Main results}\label{sec:main_results}

 We provide  upper and lower bounds for the regret in the online contract design problem. Our results  provide a nearly tight characterization for both the general case when $\mathcal{F}=[0, 1]^m$ and the linear case when $\mathcal{F}=\{\alpha \cdot \vec v \mid \alpha\in[0, 1]\}$.
For the upper bound,  we define a new complexity measure based on the covering number of a spherical code in $\mathcal{F}$, denoted as $d(\mathcal{F})$~\citep{Manin_2019,kabatiansky1978bounds}. We then show that the regret depends on $d(\mathcal{F})$, as stated in the following (informal) theorem:
\begin{theorem}[Upper bound]
For a given contract space $\mathcal{F}$, one can design algorithm $\pi$ such that
\begin{align*}
    R_T(\pi, \mathcal{F}) = O(\sqrt{m}\cdot T^{1-\frac{1}{2d(\mathcal{F})+3}}\log(T)).
\end{align*}
\end{theorem}

When $\mathcal{F}=[0, 1]^m$, one can show that  $d(\mathcal{F}) \leq m-1$. Thus we know that the upper bound for the general case is  $\widetilde O(\sqrt{m}\cdot T^{1-1/(2m+1)})$, where $\widetilde O(\cdot) $ omits logarithimic factors.  We also provide a minimax lower bound for the general case:
\begin{theorem}[Lower bound]\label{thm:informal_lower}
For any policy $\pi$, one can find some instance of the contract design problem such that
\begin{align*}
    R_T(\pi, [0, 1]^m) = \Omega( T^{1-\frac{1}{m+2}}).
\end{align*}
\end{theorem}
In summary, for the case of $\mathcal{F}=[0, 1]^m$,
we show an upper bound  $\widetilde O(\sqrt{m}\cdot T^{1-1/(2m+1)})$  and a lower bound 
$\Omega(T^{1-1/(m+2)})$.   Although the dependence on $m$ does not match exactly except for $m=1$, it shows that one has to pay an exponential cost in the number of samples in order to learn an approximately optimal contract.

Another interesting family of contracts is the set of linear contracts, $\mathcal{F}_{\mathsf{lin}}=\{\alpha \cdot \vec v \mid \alpha\in[0, 1]\}$, due to its simplicity, robustness and good approximation properties~\citep{dutting2019simple, carroll20152014}. In this case, one can show that $d(\mathcal{F}_{\mathsf{lin}}) = 0$. This naturally gives an upper bound $O(\sqrt{m}\cdot T^{2/3}\log(T))$, which matches the lower bound $\Omega(T^{2/3})$ in~\citet{ho2016adaptive} and \citet{kleinberg2003value} up to a factor of $\sqrt{m}$. By a refined analysis,  we can  get rid of the extra $\sqrt{m}$ factor and extend the upper bound to the case when the outcome size is infinite. By taking $m=1$ and adapting the argument in Theorem~\ref{thm:informal_lower}, one can also derive the same $\Omega(T^{2/3})$ lower bound   for the case of linear contracts. 
We summarize our results for linear contracts here: 
\begin{theorem}[Linear contracts]
One can find an algorithm $\pi$ such that 
\begin{align*}
   R_T(\pi, \mathcal{F}_{\mathsf{lin}}) = \widetilde \Theta( T^{2/3}).
\end{align*}
Here in $\widetilde\Theta(\cdot)$ we omit logarithmic factors. 
\end{theorem}
Here in $\widetilde \Theta( T^{2/3})$ there is no dependence on the dimension $m$.
The above  results show that the characterization of $d(\mathcal{F})$ is nearly tight in both the general case of $\mathcal{F}=[0, 1]^m$ and the linear case of $\mathcal{F}_{\mathsf{lin}}$. 

Prior work~\citep{ho2016adaptive} designs an adaptive zooming algorithm for the online contract design problem and also provides an upper bound on the sample complexity---the first such result in the literature. \citet{cohen2022learning} extends the result to  the case when the utility of the agent has bounded risk-aversion and bounds the discretization error, but under a stronger assumption of monotone-smooth. Their upper bound on sample complexity is exponential in $m$ and applies under the additional assumptions of  monotone contract and first-order stochastic dominance between the probability distributions. In comparison, we prove an exponential upper bound without these restrictions, even when the contract space is continuous. This resolves one of the open problems in~\citet{ho2016adaptive} on the structure of the contract design problem and addresses the problem of bounding the discretization error without extra assumptions. In addition, our lower bound on sample complexity is the first bound that addresses the hardness of contract design. 

We note that none of our upper bounds depend on the cardinality of the action space  $|\mathcal{A}|$. This means that our results apply to the case when the action space is infinite or continuous. However, in our construction of a lower bound in Theorem~\ref{thm:informal_lower}, we design hard instances with an action size that is exponential in $m$. Thus our lower bound does not rule out the possibility of achieving polynomial sample complexity when $|\mathcal{A}|$ is small. 

\subsection{Technical contributions}

\textbf{Upper bound.} 
The main difficulty of online contract design lies in the following two aspects:
\begin{itemize}
    \item The contract space that we search over is a continuous high-dimensional cube, $[0, 1]^m$;
    \item The utility function $u(\vec f)$ is not Lipschitz with respect to 
    $\vec f$. 
\end{itemize}
The bandit problem in metric spaces has been thoroughly studied by~\citet{kleinberg2008multi} and \citet{kleinberg2013bandits}, who provide a near-optimal discretization method when the search space is continuous and high-dimensional and the utility (reward) function is Lipschitz. However, it is unclear how one can design a good discretization method when the  utility function is not Lipschitz.

One of our main technical contributions is the design of a  new discretization method for the contract design problem. Interestingly, although the utility function $u(\vec f)$ is indeed \emph{not} Lipschitz with respect to $\vec f$, we are able to show that it is continuous along some certain directions. In particular, we provide the following structural lemma on the continuity and landscape of the utility function:
\begin{lemma}[Continuity of the utility function]\label{lem:intro_continuity}
Let $\vec \gamma = \alpha(\vec v-\vec f) + \vec r$, for some $\alpha\in(0, 1]$. For any $\vec f\in[0, 1]^m$, we have 
\begin{align*}
    u(\vec f) - u(\vec f+\vec \gamma) \leq 2 \left(\|\vec\gamma\|_\infty + \frac{\|\vec r\|_\infty}{\alpha}\right)\leq 2 \left(\|\vec\gamma\|_2 + \frac{\|\vec r\|_2}{\alpha}\right).
\end{align*}
\end{lemma}
With the help of this lemma, we can show that for a fixed $\vec f$, any other contract that lies inside a small cone with $\vec f$ as its apex will have utility close to $\vec f$. Thus it suffices to design a covering of the whole space $\mathcal{F}$ with such a cone. Thus the problem can be reduced to finding the minimum covering with a spherical code, a problem which has been studied extensively in information theory~\citep[see, e.g.,][]{Manin_2019, kabatiansky1978bounds}. We note in passing that the idea of identifying continuity directions might provide insight into general Stackelberg games and bandits with discontinuities. 

\vskip 0.1in
\noindent
\textbf{Lower bound.} Our second technical contribution is the construction of a lower bound. We employ the standard reduction in theoretical statistics from estimation to testing to establish this lower bound. We design  $\Theta ( \exp(m) ) $ instances such that each of the instances has a different optimal contract. To reduce the problem from a continuum-armed bandit problem to a finite-armed bandit problem, we divide the space into $(1/\epsilon)^m$ disjoint cubes and design actions such that each action is taken when the contract is in one of the cubes. Compared to a finite-armed bandit where the reward is decoupled for each arm, the utility in contract design is a joint function of all actions through the max operator. Thus the  reduction requires a delicate design of the production function and cost function. We again note that this construction may shed light on lower bounds for general Stackelberg games.

\subsection{Related work}

\textbf{Recent progress in contract theory.}
An extensive literature devoted to contract theory has arisen over several decades, with connections to information design~\citep[e.g.,][]{kamenica2011bayesian} and mechanism design~\citep[e.g.,][]{hurwicz2006designing}. A recent focus has been on the computational aspects of contract design. In particular,  \citet{babaioff2006combinatorial} study computational aspects of the
problem of incentivizing a collection of agents to perform a task requiring coordination among the agents. \citet{dutting2021complexity, guruganesh2021contracts, castiglioni2022designing} show that finding optimal
contracts is computationally hard when the agent can be one of several types or the outcome space is large. A complementary line of work has studied how well simple contracts can approximate the optimal utility; in particular \citet{dutting2019simple, castiglioni2021bayesian} show that the best linear contract can approximate the performance of the optimal contract well under certain conditions, and \citet{carroll20152014} shows that linear contracts are robust to unknown actions.

In addition to the work that focuses on computational issues, there has been interest in exploring linkages between statistical inference ideas and contract theory.
In particular, \citet{schorfheide2012use, schorfheide2016hold} study the problem of delegating statistical estimation to a strategic agent, and \citet{frazier2014incentivizing}
consider the problem of incentivizing learning agents in a bandit setting. Recent work by \citet{bates2022principal} shows how hypothesis testing can be carried out via the interactions between a principal and a set of agents. There is also work that applies learning methods to contract theory.  Notably, \citet{ho2016adaptive} formulate a general problem of the estimation of the distribution of agent types when a principal transacts with many agents whose values are drawn from the distribution.  \citet{cohen2022learning} extends the analysis to  the case when the utility of agent has bounded risk-aversion.  However, as is discussed in Section~\ref{sec:main_results}, they rely on several strong assumptions, in particular monotone contracts and first-order stochastic dominance among the probabilities. 

\vskip 0.1in
\noindent 
\textbf{Stackelberg games.}
Our setting is a special case of repeated Stackelberg games~\citep{von2010market, marecki2012playing, lauffer2022no}, where the principal leads and the myopic agent follows with its best response for the current round. \citet{bai2021sample} considers the repeated Stackelberg game where both the principal and the agent learn their optimal actions (a Stackelberg equilibrium) from samples. However, they assume a central controller that can determine the action of both the principal and the agent. Moreover, they rely on an assumption of a bounded gap between optimal response and $\epsilon$-approximate best response. 
In contrast, in our framework, we assume that the agent's utility is unknown, and that the agent always takes the best response. Other examples of Stackelberg games include Stackelberg security games~\citep{conitzer2006computing, tambe2011security}, strategic learning, and dynamic task pricing. We discuss the latter two in more detail.


\vskip 0.1in
\noindent
\textbf{Strategic learning and performative prediction.}
Strategic classification and regression are also special cases of general Stackelberg games~\citep{hardt2016strategic, dong2018strategic,chen2019grinding, zrnic2021leads, liu2016bandit}. It is usually assumed that the agent utility includes some cost related to the distance between the true feature of the agent and the manipulated feature. This introduces Lipschitzness into the utility function for the principal. There are also different variants of the formulation. For example, \citet{dong2018strategic, chen2019grinding} consider the setting where the data to be classified are generated by an adversary, and the observation in each round is the perturbed feature, which leads to a problem of bandit convex optimization.  \citet{hardt2016strategic, zrnic2021leads, liu2016bandit} consider stochastic settings where the data are generated from a distribution.  In contrast, in the setting of contract theory, the cost can be arbitrary and the utility function is not continuous and dependent on the probability distribution of the best response. \citet{harris2022strategic, yu2022strategic} consider a strategic instrumental variable regression and strategic MDP, under the assumption that the behavior of the agent is known.

Our setting can be viewed as a special case of performative prediction~\citep{perdomo2020performative}, where  the
prediction (the queried contract) causes a change in the distribution of the target utility. \citet{jagadeesan2022regret} study the problem of regret minimization in performative prediction. However, they assume that the change in distribution is Lipschitz with respect to the queried prediction. This does not hold in the contract design problem since the probability distribution can change drastically when we move the contract slightly. 

\vskip 0.1in
\noindent
\textbf{Dynamic task pricing and mechanism design.}
The problem of dynamic task pricing can be viewed as a special case of the online contract design in  one-dimensional space (also a special case of linear contract design). In each round, the principal offers a price $f_t$ to purchase one item. The agent agrees to sell if and only if $f_t \geq c_t$
, where $c_t\in[0, 1]$ is the agent’s private cost for
this item and is sampled from some unknown fixed distribution~\citep{kleinberg2003value, ho2016adaptive}.
Our results on linear contracts also recover the optimal rate $\widetilde \Theta(T^{2/3})$ for the problem of dynamic task pricing in~\citep{kleinberg2003value, ho2016adaptive}. See Remark~\ref{rmk:linear} for more discussions. The general theory of mechanism design and  auction also considers revenue maximization under private information with agent (buyer)~\citep{myerson1981optimal, hart2013menu, cai2017learning, daskalakis2015multi}. In the case of multi-item mechanism design, the principal will offer a menu that lists the prices for bundles of items from which the agent can choose. Compared to the contract theory, it is assumed that  both the action  and the outcome depend on the contract the principal proposes. 


\section{Problem Formulation}\label{sec:formulation}
\subsection{Single-round interaction}
Consider the problem of a single-round hidden-action contract design. The principal proposes a task that may be completed by the agent, and a contract specifying how the agent will be paid if she achieves some  certain outcome. The agent observes the contract and chooses some action, which in turn determines her own deterministic cost and the distribution of the final outcome. The agent chooses the action strategically so as to maximize her expected utility. Based on the realized outcome, the principal receives the corresponding reward and pays the agent the amount specified in the contract. We formally define the terminologies and problems as follows.

\begin{itemize}

    \item {\it Outcome:}
Let $\mathcal{O}$ be an ordered set of all possible outcomes ($\mathcal{O}$ may contain a finite or  infinite number of outcomes). Without loss of generality, we assume that $0\in \mathcal{O}\subset [0, 1]$, where $0$ is recognized as the null outcome.   
\item {\it Value:} Let $v(\cdot):\mathcal{O}\mapsto [0, 1]$ be the value function for the principal.  We assume  that $v(0)=0$ for the null outcome and $v$ is a non-decreasing function, i.e. $v(i)\leq v(j)$ for any $i,j\in\mathcal{O}$ with $i\leq j$.
\item {\it Contract}: The  principal would like to propose some contract $f:\mathcal{O}\mapsto [0, 1]$ that specifies the payment for each outcome $o\in\mathcal{O}$. We assume that the principal searches the contract among a given set $\mathcal{F}$. When $\mathcal{O}$ is a finite set with $|\mathcal{O}|=m$, the contract function can be viewed as an $m$-dimensional vector $\vec f \in \mathcal{F}\subset [0, 1]^m$.
\item {\it Agent type:} Let $\mathcal{T}$ be the set of all possible types of agents. The type of the agent is drawn randomly from an unknown distribution $\rho\in\Delta^{\mathcal{T}}$. Each type of agent may have a different production function and cost associated with the actions. We define them precisely below.
\item {\it Action:}  Let $\mathcal{A}$ be the set of all possible actions that the agents might take ($\mathcal{A}$ may contain a finite or infinite number of actions). 

\item {\it Production and Cost: } Let $p_i(| \cdot):\mathcal{A}\mapsto \Delta^{\mathcal{O}}$ and $c_i(\cdot):\mathcal{A}\mapsto [0, 1]$ be the production function and cost function for  an agent with type $i\in\mathcal{T}$. That is, for an agent with type $i$, the cost of choosing action $a\in\mathcal{A}$ is $c_i(a)$, and the outcomes of choosing action $a\in\mathcal{A}$ follow the distribution $p_i(\cdot | a)$. 
Let $0\in\mathcal{A}$ be the null action which has $0$ cost and leads to outcome $0$ deterministically. That is, one has $p_i(o=0|a=0)=1, c_i(0)=0$ for any $i\in\mathcal{T}$. 
\end{itemize}

Based on the production function $p_i$, cost $c_i$ and the given contract $f$, a strategic agent with type $i\in\mathcal{T}$ would choose action that maximizes her own expected utility:
\begin{align}
    a_i^\star(f) = \argmax_{a\in\mathcal{A}} \mathbb{E}_{X\sim p_i(\cdot|a)}[f(X)] - c_i(a).\label{eq:agent_ic}
\end{align}
The principal would like to design the contract such that her own expected utility is maximized. The utility of the principal is defined as
\begin{align}
     u(f) & = \mathbb{E}_{i\sim \rho}\mathbb{E}_{X\sim p_i(\cdot | a_i^\star(f))}[v(X) - f(X)]. \label{eq:principal_ic} 
\end{align}
When the principal is aware of the type distribution, the production function, and the cost function, the optimal contract can be found by solving the maximization problem $f^\star = \argmax_{f\in\mathcal{F}} u(f)$, which is a bi-level optimization problem due to the best response from agents.


When the outcome and action set are finite with $|\mathcal{O}|=m$, the contract  and value can be represented as an $m$-dimensional vector $\vec f\in[0,1]^m$ and $\vec v\in[0,1]^m$. For the problem with single type agent (when $\rho$ is supported on a singleton), the optimal contract can be computed in polynomial time by linear programming. When there are multiple types, the computation of the optimal single contract or optimal menu is shown to be APX-Hard in~\citet{guruganesh2021contracts}. In this paper, we focus on the statistical complexity of the dynamic model, which is introduced in the next section.

\subsection{Multi-round interaction}

We introduce  an online model where the principal and the agents interact in $T$ rounds. The principal always has the same known value $v$ over the outcomes in all the $T$ rounds.
In each round $t\leq T$, a new agent arrives with type $i^{(t)}$ following distribution $\rho$. Neither the distribution $\rho$ nor the type $i^{(t)}$ are known to the principal. The principal proposes contract $f^{(t)}$ to the new agent based on her historical observations. The agent, based on the contract, chooses her action strategically to maximize her expected utility. The outcome follows a distribution $p_{i^{(t)}}(\cdot|a)$ based on the chosen action $a$. The principal observes the realized outcome and receives a reward for the outcome. The principal knows her own reward function $v$ but does not know the type of the agent, the production function $p_i$, and the cost of agent $c_i$. The detailed interaction protocol is summarized in Procedure~\ref{procedure.contract}.
\begin{procedure_con}\label{procedure.contract}
Procedure at round $t$:
\begin{enumerate}
    \item A new agent arrives with $i^{(t)}\sim \rho$. 
    \item The principal selects a (possibly random) contract $f^{(t)}$ from the set $\mathcal{F}$ based on the value $v$ and the prior information $\mathcal{H}^{(t-1)}=(f^{(1)}, o^{(1)}, \cdots, f^{(t-1)}, o^{(t-1)})$.
    \item The agent with type $i^{(t)}$ chooses the best action $a_{i^{(t)}}^\star(f^{(t)})$ according to Equation (\ref{eq:agent_ic}).
\item The principal observes the outcome $o^{(t)}\sim p_{i_t}(\cdot | a_{i^{(t)}}^\star(f^{(t)}))$ and receives utility $v(o^{(t)}) - f^{(t)}(o^{(t)})$.
\end{enumerate}
\end{procedure_con}

A key aspect of this interaction protocol is information asymmetry. The agents make decisions based on their private type and action sets. However, both the type and actions are unknown to the principal. The principal is only able to observe the past outcomes from the agents' best responses.
In the overall procedure, the principal would like to find some good policy $\pi$ that maps the history information $\mathcal{H}^{(t-1)}$ to a new (possibly random) contract $\vec f^{(t)}$. This forms  a repeated Stackelberg game since the agent always takes best response based on the principal's action at the current round. The target is to minimize the Stackelberg regret that compares the utility achieved by the optimal single deterministic contract with the utility from following policy $\pi$, defined as
\begin{align*}
    R_T(\pi,\mathcal{F}) = \sup_{f\in\mathcal{F}}\sum_{t=1}^T \mathbb{E}_{ f^{(t)}\sim \pi(\cdot|\mathcal{H}^{(t-1)})}[u( f) - u(f^{(t)})].
\end{align*}
Note that a sublinear regret implies that a mixed strategy that is uniform over $\{f^{(t)}\}_{t=1}^T$ will approximate the optimal single contract (or Stackelberg equilibrium). Thus throughout the paper, we focus on the rate of the Stackelberg regret $R_T(\pi,\mathcal{F})$.

\section{Learning Contracts with Finite Outcomes}\label{sec:general}

We first study the sample complexity of learning an optimal contract when  the number of outcomes  is  finite. Let $|\mathcal{O}| = m$. Without loss of generality, we index the outcomes as  $\mathcal{O}=\{0, 1,\cdots, m-1\}$. In this case, the contract $f$ and the reward $v$ can be represented by  $m$-dimensional vectors $\vec f = (f_0, f_1,\cdots, f_{m-1})$, $\vec v = (v_0, v_1,\cdots, v_{m-1})$ where $f_o\in[0, 1], v_o\in[0, 1], \forall o\in\mathcal{O}$. We can rewrite the expected utility for the principal in Equation (\ref{eq:principal_ic}) as 
\begin{align}\label{eq:finite_expected_utility}
    u({\vec f}) & = \mathbb{E}_{i\sim\rho}\left[\sum_{o=0}^{m-1} p_i(o|a^\star_i(\vec f))(v_o-f_o)\right],
    \\ 
&\text{ where }   a_i^\star(\vec f) = \argmax_{a\in\mathcal{A}} \sum_{o=0}^{m-1} p_i(o|a)\cdot f_o - c_i(a), \forall i\in\mathcal{T}. \nonumber 
\end{align}
We denote the search space for contracts as $f\in \mathcal{F}\subset [0, 1]^m$. The  optimal contract is defined as $\vec f^\star = \argmax_{\vec f\in\mathcal{F}}  u(\vec f)$.

\subsection{Upper bound}\label{sec:general_upper}


We begin with a  sublinear regret upper bound based on a specific discretization of the space. We  first discretize 
the directions of all the vectors that start from the point $\vec v$ and end at any point $\vec f\in\mathcal{F}$. This is achieved by the design of  spherical code from information theory~\citep{Manin_2019, kabatiansky1978bounds}. Next, we  discretize the length of the vectors of each direction uniformly.
Concretely, for a given set of candidate contracts $\mathcal{F}$, we define the discretized set $\mathcal{S}_\epsilon(\mathcal{F})$ as follows.
\begin{definition}[Discretized set $\mathcal{V}_\epsilon$ and $\mathcal{S}_\epsilon$]\label{def:discretization}
Given a set of candidate contracts $\mathcal{F}\subset [0, 1]^m$ and the value of the principle $\vec v\in[0, 1]^m$ (not necessarily in $\mathcal{F}$), we define the discretized set $\mathcal{S}_\epsilon$ via the following two-step discretization procedure:
\begin{enumerate}
    \item {\bf Step 1: vector direction discretization.} Let $\mathcal{V}_{\epsilon}(\mathcal{F})\subset \mathcal{S}^{m-1}$ denote the set of minimum $\epsilon$-covering under the normalized inner product for the set of vectors  $\{\vec f - \vec v \mid \vec f\in\mathcal{F}\}$. Concretely, we define $\mathcal{V}_{\epsilon}(\mathcal{F})$ as  the minimum set such that for any  vector $\vec f - \vec v$ with $\vec f\in\mathcal{F}$, one can find a vector $\gamma \in \mathcal{V}_{\epsilon}(\mathcal{F})$ such that  ${\langle \vec f-\vec v, \vec \gamma\rangle}\geq \cos (\epsilon) \|\vec f-\vec v\|_2$, and $\|\gamma\|_2=1$.
    \item {\bf Step 2: vector length discretization.} We set
    \begin{align} \label{eq.discretization}
    \mathcal{S}_\epsilon(\mathcal{F}) = \{\vec v + \sqrt{m}\beta  \cdot \vec \gamma \mid \beta \in \{k\epsilon\}_{k=0}^{1/\epsilon}, \vec \gamma\in\mathcal{V}_{\epsilon^2}(\mathcal{F})\} \cap \mathcal{F}.
    \end{align}
\end{enumerate}
\end{definition}

Let $N(\mathcal{F}, \epsilon) = |V_\epsilon(\mathcal{F})|$ be the covering number of the space $\mathcal{F}$ under the normalized inner product. Similarly, we define $M(\mathcal{F}, \epsilon)$ as the maximum $\epsilon$-packing number under the same discrepancy measure, which is can be upper bounded by  the maximum cardinality of the spherical code~\citep{Manin_2019, kabatiansky1978bounds}.  We  define the intrinsic dimension of the contract design problem as the logarithmic of the covering number of the space $\mathcal{F}$ via a spherical code.  
\begin{definition}[Intrinsic dimension $d(\mathcal{F})$]\label{def.intrisic_d}
 We define the intrinsic dimension of the contract design problem using the metric entropy of the space $\mathcal{F}$:
 \begin{align*}
     d(\mathcal{F}) \coloneqq  \sup_{\epsilon<0.1}\frac{\log N(\mathcal{F},\epsilon)}{\log(1/\epsilon)}.
 \end{align*}
\end{definition}

This implies that for any $\epsilon<0.1$, we have $|\mathcal{V}_\epsilon| \leq (\frac{1}{\epsilon})^{d(\mathcal{F})}$ and 
$|\mathcal{S}_\epsilon| \leq (\frac{1}{\epsilon})^{2d(\mathcal{F})+1}$.
As a concrete example, for the largest possible contract space $\mathcal{F} = [0, 1]^m$, 
it is shown that  one has 
$M([0,1]^m, \epsilon) =  O((\frac{1}{\epsilon})^{m-1})$   when $\epsilon<0.1$ (see e.g.~\cite{386969}). Thus  we know that $ N([0,1]^m, \epsilon) = O((\frac{1}{\epsilon})^{m-1})$ since $ M([0,1]^m, 2\epsilon) \leq N([0,1]^m, \epsilon)\leq M([0,1]^m, \epsilon)$. This gives  that  $|\mathcal{S}_\epsilon([0,1]^m)|\leq O((\frac{1}{\epsilon})^{2m-1}) $. 
This upper bound also implies that $d(\mathcal{F}) \leq m-1$ for any $\mathcal{F}\subset [0, 1]^m$. Note that although the above upper bound on covering number holds information-theoretically, it is unclear whether one can design efficient algorithm to find them in high dimensions. In fact, the best known linear-programming based approach gives a covering of $ N([0,1]^m, \epsilon) = O((\frac{1}{\epsilon})^{C\cdot m})$ for some constant $C$~\citep{Manin_2019, kabatiansky1978bounds}. 

Now we are ready to present the main algorithm and result, which provide an upper bound for the Stackelberg regret based on the intrinsic dimension. 
We first discretize the space $\mathcal{F}$ into $\mathcal{S}_\epsilon(\mathcal{F})$, and then  run any worst-case optimal regret minimization algorithm for finite-armed bandit problem. We take the Upper Confidence Bond ($\mathsf{UCB}$) algorithm as one example and include the detailed pseudo-code in Appendix \S \ref{app:ucb}. 
The UCB algorithm requires an observation of the sampled reward in each round. In the contract design problem in  Procedure~\ref{procedure.contract}, the sampled reward will be $\vec v(o^{(t)}) - \vec f^{(t)}(o^{(t)})$,   with population mean $u(\vec f^{(t)})$.

\begin{algorithm}[!htbp]
	\caption{Online finite contract design}
	\label{alg:finite}
	\textbf{Input:} The contract space $\mathcal{F}\subset [0, 1]^m$, total rounds $T$.\\ 
 Set $\epsilon = T^{-1/(2d(\mathcal{F})+1)}$ and discretized set $\mathcal{S}_\epsilon(\mathcal{F})$ as in Definition~\ref{def:discretization}.\\
 Run $\mathsf{UCB}$ for $T$ rounds with action set $\mathcal{S}_\epsilon(\mathcal{F})$.  
	
\end{algorithm}
\begin{theorem}\label{thm:general_upper} 
Assume that $\epsilon<0.1$ in Algorithm~\ref{alg:finite}. Let $\pi$ be the policy produced in Algorithm~\ref{alg:finite}. For some universal constant $C$, we have 
\begin{align*}
 R_T(\pi, \mathcal{F}) \leq C\cdot \sqrt{m}\cdot T^{(2d(\mathcal{F})+2)/(2d(\mathcal{F})+3)}\cdot \log(T).
\end{align*}
\end{theorem}

The proof is deferred to Appendix \S\ref{app:proof_upper}. As a direct result of Theorem~\ref{thm:general_upper} and $d([0, 1]^m)\leq m-1$, we know that 
when $\mathcal{F}=[0, 1]^m$, one can achieve a Stackelberg regret $O(\sqrt{m}\cdot T^{1-{1}/{(2m+1)}})$. We also show in Appendix \S\ref{app:fosd} that when we have first-order stochastic dominance assumption and extra structural assumption on the contract space, one can achieve a rate of $\widetilde O(T^{1-1/(m+2)})$, which improves the constant factor in the exponent.  
A few more remarks and discussions on the upper bound are listed below.

\vskip 0.1in
\begin{remark}[Comparison with existing results]
The prior work of~\citet{ho2016adaptive} first studies the  online contract design problem. They propose an adaptive zooming algorithm that gradually zooms into the smaller $m$-dimensional cubes that may contain the optimal contract. The algorithm also provides the first upper bound that is exponential in $m$. This approach has the following limitations which we address in this work:
\begin{itemize}
    \item {\it Bounding the discretization error.}  As is discussed in~\citet{ho2016adaptive}, since the set $\mathcal{F}$ can be continuous, one known approach to obtain guarantees relative to $\sup_{\vec f\in\mathcal{F}} \mathbb{E}[u(\vec f)]$ is to start with some finite, discretized set $\mathcal{F}_\epsilon \subset \mathcal{F}$, design an algorithm with guarantees relative to $\sup_{\vec f\in\mathcal{F}_\epsilon} \mathbb{E}[u(\vec f)]$, and then, as a separate result, bound
the discretization error $$\sup_{\vec f\in\mathcal{F}} \mathbb{E}[u(\vec f)] - \sup_{\vec f\in\mathcal{F}_\epsilon} \mathbb{E}[u(\vec f)]. $$ However, they are not able to bound the generalization error. Thus their result does not apply to the case when the contract space is continuous, and cannot treat the general case when  $\mathcal{F}=[0,1]^m$. 
    \item {\it Require monotone contract.} Besides being discrete, the set $\mathcal{F}$ is further restricted to monotone contracts, $f_i\leq f_j,\forall i\leq j$, which may not contain the optimal contract.
    \item {\it Require stochastic dominance}. They make an extra assumption of first-order stochastic dominance between each of the production distributions from the actions. 
    \item 
    {\it More explicit complexity bound.} The sample complexity was expressed in terms of a quantity called \emph{virtual width}, which is hard to evaluate  in the general case  of  $\mathcal{F}=[0, 1]^m$. We introduce another notion of dimension that is explicit for this case.
    \item {\it Hardness lower bound.} The earlier work did not provide a lower bound.
\end{itemize}
\citet{cohen2022learning} extends the analysis to  the case when the utility of agent has bounded risk-aversion, and bounds the discretization error under a stronger assumption of monotone-smooth contract, which requires that $0\leq f_{i+1} - f_i \leq v_{i+1} - v_i, \forall i\in[m].$      Our work addresses all of these limitations, establishing sub-linear regret for the general continuous set $\mathcal{F}=[0, 1]^m$ without any extra assumptions. We also provide a nearly matching lower bound for   $\mathcal{F}=[0,1]^m$ in the next section. 
\end{remark}

\begin{remark}[{Implication for linear contract}] \label{rmk:linear}
In the case of linear contracts with $\mathcal{F} = \{\alpha\cdot \vec v \mid  \alpha\in[0, 1]\}$, we have $d(\mathcal{F}) = 0$ since it suffices to use a single unit vector $\vec v/\|\vec v\|_2$ to cover the directions of all the vectors that start from $\vec v$ and end at $\alpha\cdot \vec v$. This gives a rate of $O(\sqrt{m}\cdot T^{2/3})$, which matches the lower bound provided in~\citet{ho2016adaptive, kleinberg2003value} for dynamic pricing up to a factor of $\sqrt{m}$ (note that dynamic pricing is a special case of a linear contract design when there are only two outcomes: accepting or denying the item under the given price). In Section~\ref{sec:linear} we provide a tight upper bound for the linear contract case that exactly matches the lower bound and extend the result to the case when the outcome size is infinite.
\end{remark}

\begin{remark}[Examples beyond linear contract]
Our proposed complexity measure also provides new examples interpolating between the family of general contracts and linear contracts. As a concrete example, consider the family of $\mathcal{F}(\mathcal{R}) = \{\vec v + \beta\cdot \vec r \mid \beta\in[0, 1], \vec r\in\mathcal{R}\}$.    The regret upper bound is still polynomial when $|\mathcal{R}|$ is some constant, or the spherical code covering number of $\mathcal{R}$ is small. This recovers the linear contract when $\mathcal{R} = \{-\vec v\}$. 
\end{remark}

\begin{remark}
[The leading $\sqrt{m}$ factor in the bound] There is a leading $\sqrt{m}$ in the upper bound that we expect is an artifact of the proof. This is due to that in our upper bound of continuity in Lemma~\ref{lem:gen_continuity} we upper bound the $\ell_\infty$ norm with the $\ell_2$ norm. This factor can be improved if one can find a  good covering that utilizes the property of $\ell_\infty$ norm. 
\end{remark}

\begin{remark}
[The intuition behind two-step discretization] We discuss the insights behind the design of two-step discretization. As is noted in~\citet{ho2016adaptive}, the utility function $u(\cdot)$ is not continuous in general. We show, however, in Lemma~\ref{lem:intro_continuity} that it is continuous \emph{along certain directions}. In particular, for any $\vec f$ and $\vec f'$, we can bound $u(\vec f) - u(\vec f')$ if 
\begin{itemize}
    \item $\|\vec f - \vec f'\|_2\leq \epsilon$; 
    \item The angle between vector $\vec f - \vec v$ and $\vec f' - \vec v$ is upper bounded by $\epsilon$.
    \end{itemize}
This implies that for any given $\vec f$, $\vec f'$, if $\vec f'$ lies in the high-dimensional cone with $\vec f$ as its apex,  $\vec f - \vec v$ as its direction of the axis, and bounded radius and angle, then $u(\vec f) - u(\vec f')$ is bounded (see Figure~\ref{fig:covering} (a)). 

Given such a structure of continuity, one aims at covering the space $\mathcal{F}$ with the cones. To achieve this, we first stack the cones uniformly to cylinders that start from $\vec v$ as in Figure~\ref{fig:covering}. Then it suffices to cover the space with cylinders. This can be shown to be equivalent to the coverage of the angles in the sphere with unit vectors, which is referred to as a spherical code. 

\begin{figure}[!htbp]
    \centering
    \includegraphics[width=0.6\linewidth]{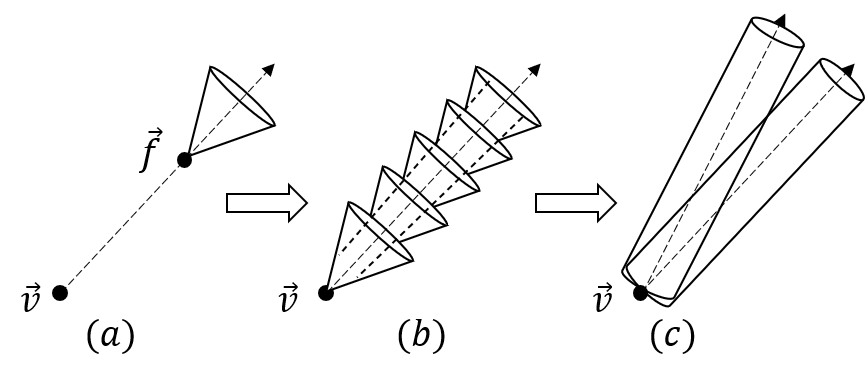}
    \caption{Construction of the covering of the space $\mathcal{F}$. (a): For any $\vec f$, one can show that there is a cone such that for any $\vec f'$ in the cone, the utility between $\vec f$ and $\vec f'$ are close to each other; (b): We stack the cone by uniformly discretizing the vector that starts from $\vec v$. This forms a cylinder (the points that are close to $\vec v$ can be handled separately). (c): The problem is then reduced to covering the space using cylinders that start from $\vec v$. }
    \label{fig:covering}
\end{figure}
\end{remark}

\subsection{Lower bound}
We  consider the lower bound for the problem. In particular, we show that the sample complexity must be exponential in $m$ in the worst case when $\mathcal{F} = [0, 1]^m$.
\begin{theorem}\label{thm:general_lower}
Let $\mathcal{I} = \{\mathcal{O}, \mathcal{A}, \vec v, c, p \mid |\mathcal{O}|=m, \vec v=(0, 1, 1,\cdots, 1)\}$ be the set of all problem instances with all agents having the same type, outcome size $m$ and principal value $1$ for all non-null outcomes. Let $u_i$ be the utility for the principal under instance $i\in\mathcal{I}$.  Then for any policy $\pi$ that produces a sequence of (possibly random and history-dependent) queries $(\vec f^{(t)})_{t=1}^T$, we have for some universal constant $C$ and $T>2^{m+2}$ that
\begin{align*}
  \inf_\pi \sup_{i\in\mathcal{I}} \sum_{t=1}^T\mathbb{E}[u_i(\vec f^\star)- u_i(\vec f^{(t)})] \geq {C\cdot T^{1-1/(m+2)}}.
\end{align*}

\end{theorem}

The proof is deferred to Appendix \S\ref{app:proof_lower}. The result shows that one has to acquire $\Omega((1/\epsilon)^{m+2})$ samples in order for the average regret $R_T(\pi, [0, 1]^m)/T$ to be as small as $\epsilon$. Thus the sample complexity must be exponential in $m$. 

We briefly discuss the intuition and construction of the lower bound. 
Although the contract space is continuous, we split the whole space into discrete regimes such that the best response within each regime remains the same. We design an exponential number of instances of the problems such that each of the instances has different optimal contracts, and furthermore, 
for any  query $\vec f\in\mathcal{F}$ and any two instances, at least one of the instances has large sub-optimality.  This necessitates the identification of the exact instance. However, we design the probability distribution such that  they are close enough under different instances  and cannot be easily distinguished by testing.

To achieve this, we first divide the space $[0, 1]^m$ uniformly to  $( 1/\epsilon)^m$ disjoint cubes. And we design $( 1/\epsilon)^m$ actions such that each action is taken when the queried contract is in one of the cubes. This approximately reduces the continuous query space to a discrete space.  Then for each instance, we choose one of the actions and lower its cost so that it will lead to a larger utility function. We show that one needs an exponential number of samples to distinguish between all the instances and find out which cube the optimal contract lies in. 

Compared  to the traditional minimax lower bound in finite-armed bandits, the main difficulty is that the observation distribution from any query $\vec f$ is coupled by the maximum operator. In traditional finite-armed bandits, querying a certain action (arm) will always return a random sample from the reward of that action no matter how the rewards are designed. In contrast, for the online contract design problem, one can only change the reward (utility) of the principal in different instances through the probability distribution and the cost of the agents' action. Moreover, the change in probability distribution or the cost will in turn change the best response of the agent. Thus when querying some $\vec f$ in two different instances, the agent may take completely different actions. 

The current   construction of lower bound does not satisfy the first-order stochastic dominance and monotone or monotone-smooth assumption.  Whether making those assumptions would lead to better regret or sample complexity bounds remains as an open question.

\section{Learning Linear Contracts with Continuous Outcomes}
\label{sec:linear}
In  Remark~\ref{rmk:linear}, we know that in the case of linear contracts, $\mathcal{F}_{\mathrm{lin}} = \{\alpha\cdot \vec v \mid  \alpha\in[0, 1]\}$, Theorem~\ref{thm:general_upper} provides an upper bound of $O(\sqrt{m}\cdot T^{2/3}\cdot \log(T))$. In this section, we improve  the result to get rid of the  extra $\sqrt{m}$ factor, and generalize the result to the case when the outcome set is continuous (infinite), indexed by $[0, 1]$. In this case, one searches the contract only in the space of  $\mathcal{F}_{\mathrm{lin}} = \{\alpha \cdot v(\cdot) | \alpha\in[0, 1]\}$, where $v$ is the fixed known reward function for the outcome. We can rewrite the expected utility for the principal in Equation (\ref{eq:principal_ic}) as
\begin{align}\label{eq:linear_expected_utility}
    u(\alpha) & = \mathbb{E}_{i\sim\rho}\left[\mathbb{E}_{X\sim p_i(\cdot|a^\star_i(\alpha))}[(1-\alpha) v(X)]\right],
    \\ 
&\text{ where }   a_i^\star(\alpha) = \argmax_{a\in\mathcal{A}} \mathbb{E}_{X\sim p_i(\cdot|a)}[\alpha \cdot v(X)] - c_i(a), \forall i\in\mathcal{T}. \nonumber 
\end{align}
Thus the  optimal single contract can be found by  $\alpha^\star = \argmax_{\alpha\in[0,1]}  u(\alpha)$.

\subsection{Upper bound}
We show  that linear contracts can be learned efficiently even when the set of outcomes and actions both have an infinite number of elements. The algorithm is exactly that for a one-dimensional continuum-armed bandit problem, which first uniformly discretizes the interval $[0, 1]$, and runs no-regret algorithms for finite-armed bandit problems. The key observation is that although $u(\alpha)$ is not continuous with respect to $\alpha$, it is indeed right-continuous, i.e. $u(\alpha) - u(\alpha + \epsilon)$ is bounded by $\epsilon$. This enables  the application of  standard continuum-armed bandit algorithms.

\begin{algorithm}[!htbp]
	\caption{Online linear contract design}
	\label{alg:linear}
	\textbf{Input:} The total rounds $T$.\\ 
	Set $\epsilon = (T/\log(T))^{-1/3}$.  Define the uniformly discretized set of parameter $\mathcal{S}_\epsilon = \{0, \epsilon, 2\epsilon, \cdots, 1 \}$. \\
     Run $\mathsf{UCB}$ for $T$ rounds with action set $\mathcal{S}_\epsilon$ and reward $u(\cdot)$. 
	
\end{algorithm}

\begin{theorem}\label{thm:linear_upper}
Let $\pi$ be the policy specified in Algorithm~\ref{alg:linear}. Then one has 
\begin{align*}
R_T(\pi, \mathcal{F}_{\mathrm{lin}})  = \sum_{t=1}^T \mathbb{E}[u(\alpha^\star) - u(\alpha^{(t)})] \leq 2 T^{2/3}\cdot \log^{1/3}(T).
\end{align*}
\end{theorem}

The proof is deferred to Appendix \S\ref{app:proof_upper_linear}. Note that although we are still searching for optimal linear contracts in an $m$-dimensional space, the family of contracts is parameterized by a scalar $\alpha$. Thus there is no dependence on $m$ in the rate. 
Compared to the Remark~\ref{rmk:linear}, the result gets rid of the extra $\sqrt{m}$ factor and achieves a near-optimal rate up to a logarithmic factor. In fact, the case of linear contracts can be viewed as a generalization of dynamic task pricing~\citep{kleinberg2003value, ho2016adaptive}. In dynamic task pricing, the principal offers a price $f_t$ to purchase one item in each round. The agent agrees to sell if and only if $f_t \geq c_t$
, where $c_t\in[0, 1]$ is the agent’s private cost for
this item and is sampled from some unknown fixed distribution. The principal receives a payoff $1-f_t$ if the item is sold, and zero otherwise.
Our results on linear contracts recover the optimal rate $\widetilde \Theta(T^{2/3})$ for the problem of dynamic task pricing.
\subsection{Lower bound}

Following a similar argument in the lower bound for the general case, we provide the following lower bound for the sample complexity of learning a linear contract. Although the lower bound $\Omega(T^{2/3})$ is already provided in \citet{ho2016adaptive, kleinberg2003value}, their construction is based on the dynamic task pricing problem with multiple types of agents, while our construction only assumes a single type of agent. 
We present the lower bound for linear contracts below.

\begin{theorem}\label{thm:linear_lower}Let $\mathcal{I} = \{\mathcal{O}, \mathcal{A}, v, c, p\mid |\mathcal{O}|=2, \vec v=(0, 1)\}$ be the set of all problem instances with agents all from the same type, two outcomes and fixed value $\vec v=(0, 1)$. Let $\mathcal{F}_{\mathsf{lin}} = \{\alpha \cdot \vec v\mid  \alpha\in[0, 1]\}$. Then for any policy, $\pi$ that produces a sequence of (possibly random and history-dependent) queries $(\alpha^{(t)})_{t=1}^T\in\mathcal{F}$, we have for some universal constant $C$ that
\begin{align*}
  \inf_\pi \sup_{i\in\mathcal{I}} \sum_{t=1}^T\mathbb{E}[u_i(\alpha^\star)- u_i(\alpha^{(t)})] \geq C\cdot T^{2/3}.
\end{align*}
\end{theorem}

The proof is deferred to Appendix \S\ref{app:proof_lower_linear}.  The lower bound matches the upper bound in Theorem~\ref{thm:linear_upper} up to logarithmic factors, which pins down the sample complexity of learning a linear contract.

\begin{remark}[Comparison with Theorem~\ref{thm:general_lower}]
The lower bound provided above matches the rate of Theorem~\ref{thm:general_lower} when taking $m=1$. 
Although they share a very similar construction, the lower bound for the linear contracts is not a direct corollary of Theorem~\ref{thm:general_lower}. 
The insight behind taking $m=1$ in Theorem~\ref{thm:general_lower} for linear contracts is that when the value is a two-dimensional vector $\vec v = (0, 1)$ (note that the value for the first null-outcome must be zero by definition), the linear contract is $\alpha\cdot \vec v = (0, \alpha)$. Thus it suffices to search for the second dimension of the contract while fixing the first dimension to zero. We formalize the argument in the proof of Theorem~\ref{thm:linear_lower}.
\end{remark}
\section{Discussion and Future Work}

We have presented a study of the sample complexity for the repeated hidden-action principal-agent model. Our results open up a number of new areas of inquiry, which we briefly review here.

\textbf{The sample complexity when the action size is small.} Although our analysis shows that  exponential in $m$ samples is both sufficient and necessary for learning general contracts, our lower bound construction is based on instances with action size that is also exponential in $m$. This does not rule out the possibility that one can learn the optimal contract with fewer samples when the action size is small. In particular, we know exactly that the utility function forms a piecewise linear function with discontinuity only in the boundary of each linear function. Thus one may infer the production distribution and cost for each action from past observations. Although learning general piecewise linear functions can be hard in high dimensions~\cite[see, e.g.,][]{dong2021provable}, the extra structure in the  contract design problem  may help improve the sample complexity when the action size is small.

\textbf{Extensions of the contract design problem.} 
There are several variants of the contract design problem. First, a 
common way of specifying a contract is to specify the payment according to the relative rankings of the agents' performance among all the agents, instead of specifying the payment for each outcome. It is unclear how the sample complexity changes when we attempt to learn the optimal contract with the relative performance. 
Second, in practice, the principal usually has partial or full information about the agent. It is interesting to see how context helps with the contract design problem. 
Third, as is the case in traditional Bayesian mechanism design (which
also has hidden types), the single deterministic contract may not be optimal when the agent has multiple hidden types. Instead, the principal may wish to present the agent with a menu of different contracts (each
specifying how much the principal would transfer for each outcome). Each agent now begins by choosing
their most desirable contract from this menu, then picking the action which maximizes their utility as a
result of this contract. It would be interesting to see how the menu complexity and sample complexity depend on the model parameters.

\textbf{Connection with continuum-armed bandits with discontinuities and general Stackelberg games.}
The online contract design problem is a special case of a continuum-armed bandit problem with discontinuous reward. The utility function in this case forms a piecewise linear function with discontinuity. It would be interesting to see similar insights applied to other bandit problems without Lipschitzness. Additionally, our results provide analysis for a special case of general Stackelberg games with myopic players. It would be interesting to see if the techniques and results can be extended to the case of general Stackelberg games.

\bibliography{ref}
\newpage
\appendix
\section{Upper Bound for Uniform Discretization under First-Order Stochastic Dominance}\label{app:fosd}

In Section~\ref{sec:general}, we establish an $O(T^{1-C/m})$ upper bound based on the spherical code. In this section, we show that one can improve the rate to $O(T^{1-1/(m+1)})$ under extra structural assumptions on the contract and the production distribution.

We introduce  the following assumptions on the production distribution and agent behavior following along the lines of~\cite{ho2016adaptive}.
\begin{assumption}[First-Order Stochastic Dominance (FOSD)]\label{asm:fosd}
Let $\mathbb{P}_i(X\geq o|a) = \int_{o'\geq o, o'\in\mathcal{O}} p_i(o'|a) do'$ be probability of arriving at outcomes that are no worse than $o$  when the agent with type $i$ takes action $a$.
For any agent type $i\in\mathcal{T}$ and any two actions $a,a'\in\mathcal{A}$, $a\neq a'$, we say that $a$ has FOSD over $a'$ for type $i$ if $\mathbb{P}_i(X\geq o|a) \geq \mathbb{P}_i(X\geq o|a'), \forall o\in\mathcal{O}$. We say that type $i$ satisfies the FOSD assumption if for any two distinct actions, one action has FOSD over the other for type $i$. We assume that all types satisfy this assumption.
\end{assumption}
To get the improved result, we also need to restrict ourselves to a specific contract family with monotone $v_o-f_o$, formally defined as below. \begin{assumption}[Monotone $v_o - f_o$]\label{asm:monotone}
 let  $\vec g = \vec v - \vec f$. We assume that the optimal contract satisfies $g_i\leq g_j, \forall i<j, i,j\in\mathcal{O}$. 
\end{assumption}
Note that \cite{ho2016adaptive} also make an assumption of 
monotone $f_o$ for the optimal contract. Both the family of monotone $v_o-f_o$ contract and monotone $f_o$ contract contain the family of linear contract. Instead of assuming that the optimal contract lies in the family, we may also change the definition of regret and only compare with the best contract inside the family.

Under both assumptions, we can show that uniform discretization gives an improved rate for learning optimal contract. The algorithms and guarantees are provided below.

\begin{algorithm}[!htbp]
	\caption{Online finite contract design with uniform discretization}
	\label{alg:finite_uniform}
	\For{$t \gets 1, \ldots, T$}{
		Set $\epsilon = (Tm^2/\log(T))^{-1/(m+2)}$. Define the uniformly discretized set of parameter $\mathcal{S}_\epsilon = \{0, \epsilon, 2\epsilon, \cdots, 1 \}^m$. \\
 Run $\mathsf{UCB}$ for $T$ rounds with action set $\mathcal{S}_\epsilon$. 
	}
\end{algorithm}
\begin{theorem}\label{thm:fosd}
Under Assumption~\ref{asm:fosd} and \ref{asm:monotone} one has 
\begin{align*}
 \sum_{t=1}^T \mathbb{E}[u(\vec f^\star) - u(\vec f^{(t)})] \leq m\cdot T^{\frac{m+1}{m+2}}\cdot \left({\log(T)}\right)^{\frac{1}{m+2}}.
\end{align*}
\end{theorem}
\begin{proof}[Proof of Theorem~\ref{thm:fosd}]
We first prove the following lemma:
\begin{lemma}\label{lem:fosd}
Under Assumption~\ref{asm:fosd}, one has \begin{align*}
    u(\vec f^\star) - u(\vec f^\star + \vec \gamma) \leq C\cdot \|\vec\gamma\|_\infty,
\end{align*}
for any $\vec \gamma$ with $0 \leq \gamma_i< \gamma_j$, $i\leq j$.
\end{lemma}
\begin{proof}[Proof of Lemma~\ref{lem:fosd} ]
We first show that under the FOSD assumption in Assumption~\ref{asm:fosd}, one has $a_i^\star(\vec f + \vec\gamma)$ FOSD $a_i^\star(\vec f)$   for any $\vec f$; i.e., it always holds that 
\begin{align*}
    \mathbb{P}_i(X\geq o|a_i^\star(\vec f)) \leq \mathbb{P}_i(X\geq o|a_i^\star(\vec f +\vec\gamma)), \forall o\in\mathcal{O}.
\end{align*}
Assume for contradiction that it does not hold for certain $i$, i.e.,
\begin{align}\label{eq:finite_contradiction_fosd}
    \mathbb{P}_i(X\geq o|a_i^\star(\vec f)) > \mathbb{P}_i(X\geq o|a_i^\star(\vec f +\vec\gamma)), \forall o\in\mathcal{O}.
\end{align}
Recall that $a_i^\star(\vec f)$ is the optimal action under contract $\vec f$. Thus one has 
\begin{align}
    \sum_o { p_i(o|a_i^\star(\vec f))}\cdot f_o - c_i(a_i^\star(\vec f)) \geq \sum_o { p_i(o|a_i^\star(\vec f+\vec \gamma))}\cdot f_o - c_i(a_i^\star(\vec f+\vec \gamma)).\label{eq:finite_fosd1}
\end{align}
This gives that 
\begin{align*}
     \sum_o { (p_i(o|a_i^\star(\vec f))} - p_i(o|a_i^\star(\vec f+\vec \gamma)))\cdot f_o  \geq c_i(a_i^\star(\vec f)) - c_i(a_i^\star(\vec f+\vec \gamma)).
\end{align*}
Similarly, from the fact that $a_i^\star(\vec f + \vec \gamma)$ is the optimal action under contract $\vec f + \vec \gamma$, we have
\begin{align*}
    \sum_o { (p_i(o|a_i^\star(\vec f + \vec \gamma))} - p_i(o|a_i^\star(\vec f)))\cdot (f_o+\gamma_o)  \geq c_i(a_i^\star(\vec f+\vec \gamma)) - c_i(a_i^\star(\vec f)).
\end{align*}

Summing up the above two equations, we arrive at
\begin{align*}
   \sum_o { (p_i(o|a_i^\star(\vec f + \vec \gamma))} - p_i(o|a_i^\star(\vec f)))\cdot \gamma_o\geq  0.
\end{align*}
If Equation (\ref{eq:finite_contradiction_fosd}) holds true, then by the definition of FOSD, we have $\sum_o { (p_i(o|a_i^\star(\vec f + \vec \gamma))} - p_i(o|a_i^\star(\vec f)))\cdot \gamma_o< 0$ since $\gamma_o$ is increasing, which contradicts with the above inequality. Thus we know that one always has  $a_i^\star(\vec f + \vec\gamma)$ FOSD  $a_i^\star(\vec f)$.

Next, we consider the difference $ u(\vec f^\star) - u(\vec f^\star + \vec \gamma)$, we have
\begin{align}
     u(\vec f^\star) - u(\vec f^\star + \vec \gamma) & = \mathbb{E}_{i\sim\rho}\left[\sum_{o} p_i(o|a^\star_i(\vec f^\star))(v_o-f_o^\star) - \sum_{o} p_i(o|a^\star_i(\vec f^\star+\vec \gamma))(v_o-f_o^\star-\gamma_o)\right] \\ & \leq \mathbb{E}_{i\sim\rho}\left[\sum_{o} \left(p_i(o|a^\star_i(\vec f^\star)) - p_i(o|a^\star_i(\vec f^\star+\vec \gamma))\right)\cdot (v_o-f_o^\star)\right] + \|\vec \gamma\|_\infty \\ & \leq \|\vec \gamma\|_\infty.
\end{align}
The last inequality comes from the fact that $a_i^\star(\vec f^\star)$ FOSD $a_i^\star(\vec f^\star + \vec\epsilon)$, and that $v_o- f^\star_o$ is monotone. This finishes the proof of the lemma.
\end{proof}
Now we proceed to the proof of the main theorem.  Let $ \vec f^{\star\prime} = \argmax_{\vec f\in\mathcal{S}} u(\vec f)$, where $\mathcal{S}$ is the discretized search space defined in Algorithm~\ref{alg:finite}.  We have
\begin{align*}
 \sum_{t=1}^T \mathbb{E}[u(\vec f^\star) - u(\vec f^{(t)})] \leq  &  \sum_{t=1}^T \mathbb{E}[u(\vec f^\star) -  u( \vec f^{\star\prime})]  + \sum_{t=1}^T\mathbb{E}[u(\vec f^{\star\prime}) -  u(\vec f^{(t)})] \\
 \leq & T\cdot (u(\vec f^\star) -  u( \vec f^{\star\prime})) + \sqrt{\frac{T\log(T)}{\epsilon^m}}.
\end{align*}
It suffices to bound the discretization error $u(\vec f^\star) -  u( \vec f^{\star\prime})$, which can be shown to be bounded by $m\epsilon$. For a given $\vec f^\star$, we can find a $\vec f'\in\mathcal{S}_\epsilon$ that is $m\epsilon$ close to $\vec f^\star$ in the following way: starting from the first coordinate, find the smallest $f_0' = k\epsilon$ for some integer $k$ such that $f_0'>\vec f^\star_0$. Set $\gamma_0 = f_0' - \vec f^\star_0$. We know that $\gamma_0\in[0, \epsilon]$. Next, we find the second coordinate with the smallest $f_1' = k\epsilon$ for some integer $k$ such that $f_1'>\vec f^\star_1+\gamma_0$. Set $\gamma_1 = f_1' - \vec f^\star_1$. We know that $\gamma_1\in(\gamma_0, 2\epsilon]$. Now repeat the procedure until the $m$-th coordinate. We can get some $\vec \gamma$ such that $\|\vec \gamma\|_\infty\leq m\epsilon$ and $\vec f^\star +\gamma\in\mathcal{S}_\epsilon$. Furthermore, we have
\begin{align*}
    u(\vec f^\star) -  u( \vec f^{\star\prime}) \leq u(\vec f^\star) -  u( \vec f^{\star} + \gamma) \leq m\epsilon.
\end{align*}
Finally, we obtain that 
\begin{align*}
 \sum_{t=1}^T \mathbb{E}[u(\vec f^\star) - u(\vec f^{(t)})] \leq Tm\epsilon + \sqrt{\frac{T\log(T)}{\epsilon^m}}.
\end{align*}
Taking $\epsilon = (Tm^2/\log(T))^{-1/(m+2)}$ gives the final result.
\end{proof}

\section{Upper Confidence Bound for Finite-Armed Bandits}\label{app:ucb}

In this section, we present the subroutine $\mathsf{UCB}$, which  achieves near-optimal worst-case regret for finite-armed bandits.

\begin{algorithm}[!htbp]
	\caption{Upper Confidence Bound}
	\label{alg:ucb}
	\textbf{Input:} Finite action set $\mathcal{A}$, total rounds $T$.\\ 
 Initialize $\hat r(a) = 1$ for all $a\in\mathcal{A}$.\\
\For{$t=1$ to $T$}{
Take action $a^{(t)} = \argmax_{a\in\mathcal{A}} \hat r(a)$, observe new reward $r^{(t)}$. \\
Set $T(a) = \sum_{s=1}^t 1(a = a^{(s)})$ for all $a$. Update 
\begin{align*}
    \hat r(a) = 
    \min\left(1, \frac{1}{T(a)}{\sum_{s=1}^t r^{(s)}1(a = a^{(s)})} + \sqrt{\frac{2\log(T)}{T(a)}}\right)
\end{align*}
}

\end{algorithm}

\section{Proof of Theorem~\ref{thm:general_upper}}\label{app:proof_upper}

\begin{proof}
The high-level construction of the upper bound goes as follows: we first show that $u(\vec f)$ is continuous along some directions in the space in Lemma~\ref{lem:gen_continuity}. As a direct result of Lemma~\ref{lem:gen_continuity}, for a given point $\vec f$, any point that is inside the cone with $\vec f$ as its apex will has its utility close to that of $\vec f$. Thus we use the cone to cover the whole space $\mathcal{F}$. We prove that under such covering $S_\epsilon$, the discretization error can be bounded.

We first prove the following lemma:
\begin{lemma}[Geometry of the observed channel induced by the best response]\label{lem:gen_inequality}
For any two contracts $\vec f, \vec f +\vec \gamma$ and any fixed $i$, one has 
\begin{align}
       \sum_o { (p_i(o|a_i^\star(\vec f + \vec \gamma))} - p_i(o|a_i^\star(\vec f)))\cdot \gamma_o\geq  0.
\end{align}
\end{lemma}
\begin{proof}
Recall that $a_i^\star(\vec f)$ is the optimal action under contract $\vec f$ and agent type $i$. Thus one has 
\begin{align}
    \sum_o { p_i(o|a_i^\star(\vec f))}\cdot f_o - c_i(a_i^\star(\vec f)) \geq \sum_o { p_i(o|a_i^\star(\vec f+\vec \gamma))}\cdot f_o - c_i(a_i^\star(\vec f+\vec \gamma)).
\end{align}
This gives that 
\begin{align*}
     \sum_o { (p_i(o|a_i^\star(\vec f))} - p_i(o|a_i^\star(\vec f+\vec \gamma)))\cdot f_o  \geq c_i(a_i^\star(\vec f)) - c_i(a_i^\star(\vec f+\vec \gamma)).
\end{align*}
Similarly, from that $a_i^\star(\vec f + \vec \gamma)$ is the optimal action under contract $\vec f + \vec \gamma$ and agent type $i$, we have
\begin{align*}
    \sum_o { (p_i(o|a_i^\star(\vec f + \vec \gamma))} - p_i(o|a_i^\star(\vec f)))\cdot (f_o+\gamma_o)  \geq c_i(a_i^\star(\vec f+\vec \gamma)) - c_i(a_i^\star(\vec f)).
\end{align*}
Summing up the above two equations, we arrive at
\begin{align*}
   \sum_o { (p_i(o|a_i^\star(\vec f + \vec \gamma))} - p_i(o|a_i^\star(\vec f)))\cdot \gamma_o\geq  0.
\end{align*}

\end{proof}

With the help of Lemma~\ref{lem:gen_inequality}, we can show that under some specific direction, the utility is approximately continuous.
\begin{lemma}[Continuity of the utility function]\label{lem:gen_continuity}
Let $\vec \gamma = \alpha(\vec v-\vec f) + \vec r$ for some $\alpha\in(0, 1]$. For any $f\in[0, 1]^m$, we have 
\begin{align*}
    u(\vec f) - u(\vec f+\vec \gamma) \leq 2 \left(\|\vec\gamma\|_\infty + \frac{\|\vec r\|_\infty}{\alpha}\right)\leq 2\left( \|\vec\gamma\|_2 + \frac{\|\vec r\|_2}{\alpha}\right).
\end{align*}
\end{lemma}
\begin{proof}
From Lemma~\ref{lem:gen_inequality}, we know that for $\vec \gamma = \alpha(\vec v-\vec f) + \vec r$, we have 
\begin{align}
   \sum_o { (p_i(o|a_i^\star(\vec f + \vec \gamma))} - p_i(o|a_i^\star(\vec f)))\cdot (\alpha( v_o- f_o) + r_o)\geq  0.
\end{align}
Rearranging the above formula gives us  that 
\begin{align}\label{eqn:rearranged_finite}
   \sum_o { (p_i(o|a_i^\star(\vec f ))} - p_i(o|a_i^\star(\vec f+ \vec \gamma)))\cdot ( v_o- f_o) \leq   \sum_o { (p_i(o|a_i^\star(\vec f + \vec \gamma))} - p_i(o|a_i^\star(\vec f)))\cdot  \frac{r_o}{\alpha}\leq\frac{2\|\vec r\|_\infty}{\alpha}.
\end{align}
We have
\begin{align}
     u(\vec f) - u(\vec f + \vec \gamma) & = \mathbb{E}_{i\sim\rho}\left[\sum_{o} p_i(o|a^\star_i(\vec f))(v_o-f_o) - \sum_{o} p_i(o|a^\star_i(\vec f+\vec \gamma))(v_o-f_o-\gamma_o)\right] \\ & \leq \mathbb{E}_{i\sim\rho}\left[\sum_{o} \left(p_i(o|a^\star_i(\vec f)) - p_i(o|a^\star_i(\vec f+\vec \gamma))\right)\cdot (v_o-f_o)\right] + 2\|\vec \gamma\|_\infty \\ & \leq 2\left(\frac{\|\vec r\|_\infty}{\alpha} + \|\vec \gamma\|_\infty\right).
\end{align}
The last inequality comes from Equation (\ref{eqn:rearranged_finite}). This finishes the proof of the lemma.
\end{proof}
Now we proceed to the proof of the main theorem.  Let $ \vec f' = \argmax_{\vec f\in\mathcal{S}_\epsilon} u(\vec f)$, where $\mathcal{S}_\epsilon$ is the discretized search space in Definition~\ref{def:discretization}.   We have
\begin{align*}
 \sum_{t=1}^T \mathbb{E}[u(\vec f^\star) - u(\vec f^{(t)})] \leq &  \sum_{t=1}^T \mathbb{E}[u(\vec f^\star) -  u( \vec f^{\star\prime})]  + \sum_{t=1}^T\mathbb{E}[u(\vec f^{\star\prime}) -  u(\vec f^{(t)})] 
 \\
 \leq & T\cdot (u(\vec f^\star) -  u( \vec f^{\star\prime})) + O\left(\sqrt{\frac{T\log(T)}{\epsilon^{2d(\mathcal{F})+1}}}\right).
\end{align*}
Here the last inequality is due to the guarantee for UCB algorithm. See e.g.~\citet{lattimore2020bandit}.
Now we bound the discretization error in the first term, we aim to show that for $\vec f^\star\in\mathcal{F}$, one can find one vector $\vec f^\star + \vec \gamma\in\mathcal{S}_\epsilon$ such that $ \vec \gamma = \alpha(\vec v-\vec f^\star) + \vec r$ with $\frac{\|\vec r\|_2}{\alpha} + \|\vec \gamma\|_2\leq C \epsilon \sqrt{m}$. 
Intuitively, this is correct due to that $\mathcal{S}_\epsilon$
is a valid $\epsilon\sqrt{m}$-covering of the space $[0, 1]^m$ under the distance $\|\vec \gamma \|_2 + \|\vec r\|_2/\alpha$. Below we make this argument rigorous.

First, consider the case that $\|\vec f^\star-\vec v\|_2 \leq 10\epsilon\sqrt{m}$. We have
\begin{align*}
u(\vec f^\star) =    u(\vec f^\star) - u (\vec v) = u(\vec f^\star) - u (\vec f^\star + (\vec v - \vec f^\star))\leq 2\|\vec v - \vec f^\star\|_2 \leq 20\epsilon\sqrt{m}.
\end{align*}
The second last inequality comes from Lemma~\ref{lem:gen_continuity}. This shows that when $\|\vec f^\star-\vec v\|_2 \leq 10\epsilon\sqrt{m}$, the discretization error is at most $20\epsilon\sqrt{m}$ since the optimal value is bounded by $20\epsilon\sqrt{m}$.

Next, consider the case when $\|\vec f^\star-\vec v\|_2 > 10\epsilon\sqrt{m}$.  By the definition of $\mathcal{V}_\epsilon$ in Definition~\ref{def:discretization}, we know that there always exists some unit vector $\vec \eta  \in\mathcal{V}_{\epsilon^2}(\mathcal{F})$, such that $\frac{\langle \vec \eta, \vec f^\star - \vec v\rangle}{ \| \vec f^\star - \vec v\|_2}\geq \cos(\epsilon^2)$. Now we plot the two-dimensional space spanned by vectors $\vec \eta$ and $\vec f^\star - \vec v$ as in Figure~\ref{fig:2dplot} (where we let both $\vec \eta$ and $\vec f^\star-\vec v$ originates from $\vec v$). We make a line that crosses point $E=\vec f^\star$ and perpendicular to $\vec \eta$. Let its intersection with $\vec \eta$ be the point $A$. Since $\|\vec f^\star - \vec v \|_2 >10\epsilon\sqrt{m}$, we know that $\|\overrightarrow{DA}\|_2> 10\epsilon\sqrt{m}\cos(\epsilon^2)>5\epsilon\sqrt{m}$. From the definition of $\mathcal{S}_\epsilon$, we know that one can find some point $B$ on line $DA$ such that $\|\overrightarrow{AB}\|_2\in (4\epsilon\sqrt{m}, 5\epsilon\sqrt{m})$, and that $\vec v+\overrightarrow{DB} \in \mathcal{S}_\epsilon$. Now we set $\vec \gamma = \overrightarrow{EB}$. We know that $\|\vec \gamma\|_2< \sqrt{\epsilon^4m+25\epsilon^2 m} <6\epsilon\sqrt{m}$. By setting $\vec \gamma = \overrightarrow{EC} + \overrightarrow{CB} = \alpha (\vec v - \vec f^\star) + \vec r$, we know that
\begin{align*}
    \frac{\|\vec r\|_2}{\alpha} & =  \frac{\|\vec r\|_2}{\alpha\|\vec v -\vec f^\star\|_2} \cdot  \|\vec v -\vec f^\star\|_2 = \frac{\|\overrightarrow{CB}\|_2}{\|\overrightarrow{EC}\|_2} \cdot  \|\vec v -\vec f^\star\|_2 \leq \frac{\|\overrightarrow{CB}\|_2}{\|\overrightarrow{EC}\|_2} \cdot \sqrt{m} = \sqrt{m}\tan \angle BEC< \sqrt{m}\tan \angle ABE \\ 
    & = \frac{\|\overrightarrow{AE}\|_2}{\|\overrightarrow{AB}\|_2} \cdot \sqrt{m} < \frac{\epsilon^2\sqrt{m}}{4\epsilon\sqrt{m}}\cdot \sqrt{m} < \epsilon\sqrt{m}.
\end{align*}


Combining the arguments together, we have
\begin{align*}
  u(\vec f^\star) - u (\vec f^\star +\vec \gamma) \leq 20\epsilon\sqrt{m}.
\end{align*}
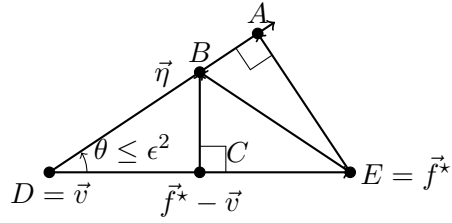
\begin{figure}[!htbp]
\centering
  \begin{tikzpicture}
    \coordinate (origo) at (0,0);
    \coordinate (f) at (4,0);
    \coordinate (A) at (2.77, 1.847);
    \coordinate (B) at (2, 1.333);
    \coordinate (C) at (2, 0);
\filldraw[black] (0,0) circle (2pt) node[below]{$D = \vec v$};
\filldraw[black] (f) circle (2pt) node[black,right] {$E = \vec f^\star$};
\filldraw[black] (A) circle (2pt) node[black,above] {$A$};
\filldraw[black] (B) circle (2pt) node[black,above] {$B$};
\filldraw[black] (C) circle (2pt) node[black,above, xshift=0.5cm] {$C$};
    \draw[thick,black,->] (origo) -- node[black, below] {  $\vec f^\star - \vec{v} $} ++(f);
    \draw[thick,black,->] (origo) --  node[black, above] {$\vec{\eta}$} ++(3,2) node (eta) [black,below] {};
    \draw[thick,black,->] (f) -- (A) {};
        \draw[thick,black,->] (f) -- (B) {};
            \draw[thick,black,->] (C) -- (B) {};

\draw [black,right angle symbol={A}{B}{f}];
\draw [black,right angle symbol={C}{B}{f}];

    \pic [draw, ->, "$\theta\leq \epsilon^2$", angle eccentricity=2.3] {angle = f--origo--eta};
  \end{tikzpicture}
  \caption{The 2 dimensional plot of the plane spanned by vector $\vec \eta$ and $\vec v - \vec f^\star$.}
\label{fig:2dplot}
\end{figure}

Finally, we have that 
\begin{align*}
 \sum_{t=1}^T \mathbb{E}[u(\vec f^\star) - u(\vec f^{(t)})] \leq O\left(T\epsilon \sqrt{m}+ \sqrt{\frac{T\log(T)}{\epsilon^{2d(\mathcal{F})+1}}}\right),
\end{align*}
and taking $\epsilon = T^{-1/(2d(\mathcal{F})+3)}$  gives the final result.

\end{proof}

\section{Proof of Theorem~\ref{thm:general_lower}}\label{app:proof_lower}

\begin{proof}
Intuitively, we would like to employ the general idea of the multiple hypothesis testing to prove the lower bound. We design exponential number of instances such that each of the instance has different optimal contracts.  Furthermore, 
for any $T$ queries $(\vec f^{(1)},\cdots, \vec f^{(T)})$, one can always find two instances such that the observation distribution is close enough such that they cannot be easily distinguished by testing, while the average regret of the two instances is large. 

To achieve this, we first divide the space $[0, 1]^m$ uniformly to  $(\lfloor 1/\epsilon\rfloor)^m$ disjoint cubes. We design $(\lfloor 1/2\epsilon\rfloor)^m$ actions such that each action is taken when the queried contract is in one of the cube. This approximately reduces the continuum-armed bandit problem to a discrete-armed bandit problem.  Then for each instance, we choose one of the action and lower its cost so that it will be chosen in a slightly larger regime, and it will lead to a larger utility function. We show that one needs exponential number of samples to distinguish between all the instances and find out where the optimal contract lies in.

Formally,
for all the problem instances, we consider agents of the same type and fix $\vec v = (0, 1, 1, \cdots, 1)$, which is always $1$ for any non-null outcome.   Let each action be indexed by a vector of integer $a(k_1, k_2,\cdots, k_m)$ with  $k_i\in\{0, 1,\cdots, \left\lfloor\frac{1}{2\epsilon}\right\rfloor-1\}$, whose production probability and cost are defined as follows
\begin{align*}
    p(k_1, k_2,\cdots, k_m) &= \left(1-\sum_{i=1}^m\frac{1}{2(1-k_i\epsilon)m}, \frac{1}{2(1-k_1\epsilon)m}, \frac{1}{2(1-k_2\epsilon)m}, \cdots, \frac{1}{2(1-k_m\epsilon)m}\right),\\
    c(k_1,k_2,\cdots, k_m) &= \frac{1}{2m} \sum_{i=1}^m \left(\frac{k_i\epsilon}{1-k_i\epsilon} - \sum_{j=0}^{k_i-1} \frac{\epsilon}{1-j\epsilon}\right).
\end{align*}
We assume that $\epsilon<0.1$ throughout the proof.
To choose action $a(k_1,k_2,\cdots, k_m)$ as the best response, the following must hold true for any $(k_1', k_2',\cdots, k_m')\neq (k_1, k_2,\cdots, k_m)$:
\begin{align*}
    p(k_1, k_2,\cdots, k_m)\cdot \vec f - c(k_1, k_2,\cdots, k_m)\geq p(k_1', k_2',\cdots, k_m')\cdot \vec f - c(k_1', k_2',\cdots, k_m').
\end{align*}
Taking the difference of both sides, notice that we can decompose the difference into the difference in each coordinate due to the design of $p$ and $c$:
\begin{align*}
    & p(k_1, k_2,\cdots, k_m)\cdot \vec f - c(k_1, k_2,\cdots, k_m) - ( p(k_1', k_2',\cdots, k_m')\cdot \vec f - c(k_1', k_2',\cdots, k_m'))  \\
    =&\sum_{l=1}^m     p(k_1, k_2,\cdots,k_l, \cdots,  k_m)\cdot \vec f - c(k_1, k_2,\cdots,k_l, \cdots,  k_m) \\ 
    & \quad - (p(k_1, k_2,\cdots,k_l', \cdots,  k_m)\cdot \vec f - c(k_1, k_2,\cdots,k_l', \cdots,  k_m)).
\end{align*}
This means that for some action to be chosen, it suffices to guarantee that all other actions that only differ in one coordinate have smaller payoff.
By enumerating all possibilities, we can see that action $a(k_1, \cdots, k_m)$ is chosen if and only if the following holds true for all coordinate $i$:
\begin{align*}
    k_i\epsilon\leq f_i-f_0 < (k_i+1)\epsilon, \text{ when } k_i < \left\lfloor\frac{1}{2\epsilon}\right\rfloor-1; \text{ or } k_i\epsilon\leq f_i-f_0 \text{ when } k_i = \left\lfloor\frac{1}{2\epsilon}\right\rfloor-1.
\end{align*}
The maximum payoff achieved by principal is $1/2$ when $f_i = k_i\epsilon, f_0 = 0$. 

Now consider  $(\lfloor 1/(8\epsilon)\rfloor)^m$ instances, indexed by $(l_1, l_2,\cdots, l_m)$ with $ l_i\in[2, 4, 6,\cdots, 2\times \lfloor 1/(8\epsilon) \rfloor]$ for all $i\in[m]$. For each of the instances, we keep the same probability distribution for each arm as $p(k_1, k_2,\cdots, k_m)$, but change the cost to the following:
\begin{align*}
    c^{(l_1,\cdots, l_m)}(k_1, k_2, \cdots, k_m) = \begin{cases}  
     c(k_1, k_2,\cdots, k_m) - \frac{\epsilon^2}{10m}, & (k_1, k_2,\cdots, k_m) = (l_1,l_2,\cdots, l_m) \\ 
    c(k_1, k_2,\cdots, k_m) - \frac{\epsilon^2}{20m}, & (k_1, k_2,\cdots, k_m) = (2 \lfloor \frac{1}{8\epsilon} \rfloor + 2,2\lfloor \frac{1}{8\epsilon} \rfloor + 2,\cdots, 2\lfloor \frac{1}{8\epsilon} \rfloor + 2) \\
    c(k_1, k_2,\cdots, k_m), & \text{otherwise} \\ 
    \end{cases}
\end{align*}
One can verify that the costs are non-negative. For the instance $(l_1, l_2,\cdots, l_m)$, the contracts for which action $a(l_1, l_2,\cdots, l_m)$ is chosen as the best response are those
inside the following set: \begin{align*}
\mathcal{C}_{(l_1,\cdots,l_m)} = \left\{\vec f \mid 
    l_i\epsilon - \frac{(1-l_i\epsilon)(1-(l_i-1)\epsilon)\epsilon}{5}\leq f_i-f_0<  (l_i+1)\epsilon + \frac{(1-l_i\epsilon)(1-(l_i+1)\epsilon)\epsilon}{5}, 1\leq i\leq m\right\}.
\end{align*}
The maximum payoff achieved is  $1/2+\sum_i  (1-(l_i-1)\epsilon)\epsilon/10m\geq 1/2 + {\epsilon}/{20}$ when $f_0=0, f_i = l_i\epsilon - \frac{(1-l_i\epsilon)(1-(l_i-1)\epsilon)\epsilon}{5}$.

Now consider another instance indexed by $(2\times \lfloor \frac{1}{8\epsilon} \rfloor + 2, \cdots, 2\times \lfloor \frac{1}{8\epsilon} \rfloor + 2)$. We keep the same probability distribution for each arm as $p(k_1, k_2,\cdots, k_m)$, but change the cost to the following:
\begin{align*}
    c^{(2\times \lfloor \frac{1}{8\epsilon} \rfloor + 2,\cdots, 2\times \lfloor \frac{1}{8\epsilon} \rfloor + 2)}(k_1, k_2, \cdots, k_m) = \begin{cases}  
    c(k_1, k_2,\cdots, k_m) - \frac{\epsilon^2}{20m}, & (k_1, k_2,\cdots, k_m) = (2\times \lfloor \frac{1}{8\epsilon} \rfloor + 2,\cdots, 2\times \lfloor \frac{1}{8\epsilon} \rfloor + 2) \\
    c(k_1, k_2,\cdots, k_m), & \text{otherwise} \\ 
    \end{cases}
\end{align*}
For this instance, the maximum payoff is $1/2 + \sum_i  (1-(2\times \lfloor \frac{1}{8\epsilon} \rfloor +1)\epsilon)\epsilon/20m\in [1/2+{\epsilon}/{30},  1/2+{\epsilon}/{25}]$ when  $f_0=0, f_i = (2\times \lfloor \frac{1}{8\epsilon} \rfloor +2)\epsilon - \frac{(1-(2\times \lfloor \frac{1}{8\epsilon} \rfloor +2)\epsilon)(1-(2\times \lfloor \frac{1}{8\epsilon} \rfloor +1)\epsilon)\epsilon}{10}, \forall i\geq 1$. Moreover, the contracts that action $a(2\times \lfloor \frac{1}{8\epsilon} \rfloor + 2, \cdots, 2\times \lfloor \frac{1}{8\epsilon} \rfloor + 2)$ is chosen as the best response are those
inside the following set: 
\begin{align*}
\mathcal{C}_{(2\times \lfloor \frac{1}{8\epsilon} \rfloor + 2, \cdots, 2\times \lfloor \frac{1}{8\epsilon} \rfloor + 2)} =& \Big\{\vec f \mid 
   (2\times \lfloor \frac{1}{8\epsilon} \rfloor +2)\epsilon - \frac{(1-(2\times \lfloor \frac{1}{8\epsilon} \rfloor +2)\epsilon)(1-(2\times \lfloor \frac{1}{8\epsilon} \rfloor +1)\epsilon)\epsilon}{20}\leq f_i-f_0, \\ 
   & f_i-f_0 <  (2\times \lfloor \frac{1}{8\epsilon} \rfloor +3)\epsilon + \frac{(1-(2\times \lfloor \frac{1}{8\epsilon} \rfloor +3)\epsilon)(1-(2\times \lfloor \frac{1}{8\epsilon} \rfloor +2)\epsilon)\epsilon}{20}, 1\leq i\leq m \Big\}.
\end{align*}

Let $\mathcal{I}'$ be the set of instances $(l_1,\cdots, l_m)$ with $1\leq l_i\leq 2\times \lfloor \frac{1}{8\epsilon} \rfloor$, $\mathcal{I}''$  be the set of $\mathcal{I}'$ unioned with  an extra instance $(2\times \lfloor \frac{1}{8\epsilon} \rfloor + 2,\cdots 2\times \lfloor \frac{1}{8\epsilon} \rfloor + 2)$.  
Let $u^{(l_1,\cdots, l_m)}$ denote the payoff for the principal under instance $(l_1,\cdots, l_m)$. 
For each of the instances, we have the following properties:
\begin{itemize}
\item The sets $\mathcal{C}_{(l_1,\cdots,l_m)}$ are disjoint for different pairs of $(l_1,\cdots,l_m)\in\mathcal{I}''$.
    \item For any instance $(l_1,\cdots, l_m)\in\mathcal{I}''$, the action $a(l_1,\cdots, l_m)$ is the only action that leads to optimal payoff for the principal. The optimal contract for instance $(l_1,\cdots,l_m)$ lies in the set $\mathcal{C}_{(l_1,\cdots,l_m)}$.
    \item  For any $(l_1,\cdots, l_m)\in\mathcal{I}''$ and $\vec f\not \in\mathcal{C}_{(l_1,\cdots, l_m)}$, the suboptimality $u^{(l_1,\cdots, l_m)}(\vec f^\star) - u^{(l_1,\cdots, l_m)}(\vec f)$ is at least $\epsilon/100.$
\end{itemize}


Now we apply Le Cam's Lemma to show the lower bound. 
From the reduction of  estimation to testing (see, e.g., Prop 15.1 of~\cite{wainwright2019high}),
we know that
\begin{align}
    \sup_{i\in \mathcal{I}} \sum_{t=1}^T\mathbb{E}[u_i(\vec f^\star)- u_i(\vec f^{(t)})]\geq  &  \sup_{(l_1,\cdots,l_m)\in\mathcal{I}''} \sum_{t=1}^T\mathbb{E}[u^{(l_1,\cdots,l_m)}(\vec f^\star)- u^{(l_1,\cdots,l_m)}(\vec f^{(t)})] \nonumber  \\ 
    \geq & \frac{1}{|\mathcal{I}''|}\sum_{(l_1,\cdots,l_m)\in\mathcal{I}''} \sum_{t=1}^T\mathbb{E}[u^{(l_1,\cdots,l_m)}(\vec f^\star)- u^{(l_1,\cdots,l_m)}(\vec f^{(t)})]  \nonumber \\ 
    \geq & \frac{\epsilon}{100}\cdot \inf_{\Psi}\sum_{t=1}^T \frac{1}{|\mathcal{I}''|} \sum_{(l_1,\cdots,l_m)\in\mathcal{I}''} \mathbb{P}_t^{(l_1,\cdots, l_m)}(\Psi\not =(l_1,\cdots,l_m)),\label{eq.testing_lower}
\end{align}
where the infimum is taken over all measurable tests $\Psi$ that maps the history observation to the indices $\mathcal{I}''$. $\mathbb{P}_t^{(l_1,\cdots, l_m)}$ is the distribution of the observations up to time $t$ under the instance $(l_1,\cdots,l_m)$. 
Now from the tree-based lower bound (see Lemma 3 in~\cite{gao2019batched}),
we have
\begin{align}
 & \inf_{\Psi} \frac{1}{|\mathcal{I}''|} \sum_{(l_1,\cdots,l_m)\in\mathcal{I}'} \mathbb{P}_t^{(l_1,\cdots, l_m)}(\Psi\not =(l_1,\cdots,l_m)) \nonumber \\
 & \geq   \frac{1}{|\mathcal{I}''|}\sum_{(l_1,\cdots,l_m)\in\mathcal{I}'} (1-\|\mathbb{P}_t^{(2\times \lfloor \frac{1}{8\epsilon} \rfloor + 2,\cdots, 2\times \lfloor \frac{1}{8\epsilon} \rfloor + 2)} -\mathbb{P}_t^{(l_1,\cdots,l_m)} \|_{\mathsf{TV}})\nonumber  \\ 
     & \geq \frac{1}{|\mathcal{I}''|}\cdot \sum_{(l_1,\cdots,l_m)\in\mathcal{I}'} \exp(-D_{\mathsf{KL}}(\mathbb{P}_t^{(2\times \lfloor \frac{1}{8\epsilon} \rfloor + 2,\cdots, 2\times \lfloor \frac{1}{8\epsilon} \rfloor + 2)}, \mathbb{P}_t^{(l_1,\cdots,l_m)} )).\label{eq.tree_lower}
     \end{align}
Now we further upper bound the KL divergence
$D_{\mathsf{KL}}(\mathbb{P}_t^{(2\times \lfloor \frac{1}{8\epsilon} \rfloor + 2,\cdots, 2\times \lfloor \frac{1}{8\epsilon} \rfloor + 2)}, \mathbb{P}_t^{(l_1,\cdots,l_m)} )$. Let $T(l_1,\cdots,l_m)$ be the number of pulls of the regime $\mathcal{C}_{(l_1,\cdots,l_m)}$ anterior to the current batch of $t$ under instance $(2\times \lfloor \frac{1}{8\epsilon} \rfloor + 2,\cdots, 2\times \lfloor \frac{1}{8\epsilon} \rfloor + 2)$. Note that the distributions of outcome in the two instances are always the same  when we query the contract that lies outside  $ \mathcal{C}_{(l_1,\cdots,l_m)}$. Moreover, we have $\sum_{(l_1,\cdots,l_m)} T_t(l_1,\cdots,l_m)\leq t$ since  the regimes $\mathcal{C}_{(l_1,\cdots,l_m)}$ do not overlap with each other and the total time step is $t$.  By the divergence decomposition lemma\footnote{Although the divergence decomposition lemma is for the finite-arm bandit problem, one can reduce our problem to finite-arm bandit problem since there are in total $(\lfloor 1/2\epsilon\rfloor) +1$ different observation distributions.}~\citep{lattimore2020bandit}, we have
    \begin{align}
   \lefteqn{\frac{1}{|\mathcal{I}''|}\cdot \sum_{(l_1,\cdots,l_m)} \exp(-D_{\mathsf{KL}}(\mathbb{P}_t^{(2\times \lfloor \frac{1}{8\epsilon} \rfloor + 2,\cdots, 2\times \lfloor \frac{1}{8\epsilon} \rfloor + 2)}, \mathbb{P}_t^{(l_1,\cdots,l_m)} ))} \nonumber \\ 
    & \stackrel{(i)}{\geq}  \frac{1}{|\mathcal{I}''|}\cdot \sum_{(l_1,\cdots,l_m)} \exp(-\sup_{j_i\in\{-1,0,1\}}D_{\mathsf{KL}}(p(l_1+j_1,\cdots,l_m+j_m), p(l_1,\cdots,l_m)) \mathbb{E}[T(l_1,\cdots, l_m)]) \nonumber \\  
    & \stackrel{(ii)}{\geq}   \frac{1}{|\mathcal{I}''|}\cdot \sum_{(l_1,\cdots,l_m)} \exp(-\sup_{j_i\in\{-1,0,1\}} D_{\mathsf{\mathcal{X}^2}}(p(l_1+j_1,\cdots,l_m+j_m), p(l_1,\cdots,l_m)) \mathbb{E}[T(l_1,\cdots, l_m)]) \nonumber \\ 
     & \stackrel{(iii)}{\geq}    \frac{1}{|\mathcal{I}''|}\cdot \sum_{(l_1,\cdots,l_m)} \exp(-C\cdot \epsilon^2 \mathbb{E}[T(l_1,\cdots, l_m)])  \nonumber \\ 
     & \stackrel{(iv)}{\geq}   \exp\left(-\frac{C}{|\mathcal{I}''|}\cdot \sum_{(l_1,\cdots,l_m)}\epsilon^2 \mathbb{E}[T(l_1,\cdots, l_m)]\right) \nonumber  \\ 
     & \stackrel{(v)}{\geq}  \exp(-C T\cdot (8\epsilon)^m\cdot \epsilon^2).\label{eq.KL_lower}
\end{align}
Here (i) is from that the observation distributions are the same under both instances when the queried contract is not in $\mathcal{C}(l_1,\cdots,l_m)$; when the queried contract is in $\mathcal{C}_{(l_1,\cdots,l_m)}$, the observed distribution  under the instance $(2\times \lfloor \frac{1}{8\epsilon} \rfloor + 2,\cdots, 2\times \lfloor \frac{1}{8\epsilon} \rfloor + 2)$  is one of the  neighborhood distributions in $p(l_1+j_1,\cdots,l_m+j_m)$, where $j_i\in\{-1,0,1\}$, while the distribution remains $p(l_1,\cdots, l_m)$ under the instance $(l_1,\cdots,l_m)$. (ii) is from that  $D_{\mathsf{KL}}$ is upper bounded by $D_{\mathcal{X}^2}$. (iii) is from the plug-in of two distributions into the $\mathcal{X}^2$-divergence. (iv) is from Jensen's inequality. (v) is from the fact that $\sum_{(l_1,\cdots,l_m)} \mathbb{E}[T_t(l_1,\cdots,l_m)]\leq t\leq T$.
Combining Equation (\ref{eq.testing_lower}), (\ref{eq.tree_lower}) and (\ref{eq.KL_lower}) gives that
\begin{align*}
      \sup_{i\in \mathcal{I}} \sum_{t=1}^T\mathbb{E}[u_i(\vec f^\star)- u_i(\vec f^{(t)})] \geq \frac{\epsilon}{100} \cdot T \cdot \exp(-C T\cdot (8\epsilon)^m\cdot \epsilon^2).
\end{align*}
By taking $\epsilon = \frac{T^{-1/(m+2)}}{8}$ we obtain
\begin{align*}
     \min_{\vec f^{(t)}} \sup_{\mathcal{I}} \sum_{t=1}^T\mathbb{E}[u(\vec f^\star)- u(\vec f^{(t)})]& \geq {C\cdot T^{1-1/(m+2)}}.
     \end{align*}
\end{proof}

\section{Proof of Theorem~\ref{thm:linear_upper}}\label{app:proof_upper_linear}

\begin{proof}[Proof of Theorem~\ref{thm:linear_upper}]
Before we prove the main result,
we first prove the following lemma, which shows that $u(\alpha)/(1-\alpha)$ is a non-decreasing  function of $\alpha$.
\begin{lemma}\label{lem:linear_property}
For any $0\leq \alpha <\alpha'\leq 1$, one has $u(\alpha)/(1-\alpha) \leq u(\alpha')/(1-\alpha')$
\end{lemma}
\begin{proof}
First, note that  $u(\alpha)/(1-\alpha) = \mathbb{E}_{i\sim\rho}\left[\mathbb{E}_{X\sim p_i(a^\star_i(\alpha))}[v(X)]\right]$. Thus the lemma is equivalent to showing that 
$\mathbb{E}_{i\sim\rho}\left[\mathbb{E}_{X\sim p_i(a^\star_i(\alpha))}[v(X)]\right] \leq \mathbb{E}_{i\sim\rho}\left[\mathbb{E}_{X\sim p_i(a^\star_i(\alpha'))}[v(X)]\right]$. We prove a stronger version of it, which is for any $i\in\mathcal{T}$, $\mathbb{E}_{X\sim p_i(a^\star_i(\alpha))}[v(X)] \leq \mathbb{E}_{X\sim p_i(a^\star_i(\alpha'))}[v(X)]$. 


Recall that $a_i^\star(\alpha)$ is the optimal action under contract $\alpha\cdot v(\cdot)$. Thus one has 
\begin{align}
    \mathbb{E}_{X\sim p_i(a_i^\star(\alpha))}[\alpha \cdot v(X)] - c_i(a_i^\star(\alpha)) \geq \mathbb{E}_{X\sim p_i(a_i^\star(\alpha'))}[\alpha \cdot v(X)] - c_i(a_i^\star(\alpha')).\label{eq:linear_fosd1}
\end{align}
This gives that 
\begin{align*}
     \alpha \cdot \left(\mathbb{E}_{X\sim p_i(a_i^\star(\alpha))}[ v(X)] -\mathbb{E}_{X\sim p_i(a_i^\star(\alpha'))}[v(X)] \right) \geq c_i(a_i^\star(\alpha))   - c_i(a_i^\star(\alpha')).
\end{align*}
Similarly, from that $a_i^\star(\alpha')$ is the optimal action under contract $\alpha'\cdot v(\cdot)$, we have
\begin{align*}
     \alpha' \cdot \left(\mathbb{E}_{X\sim p_i(a_i^\star(\alpha'))}[ v(X)] -\mathbb{E}_{X\sim p_i(a_i^\star(\alpha))}[v(X)] \right) \geq c_i(a_i^\star(\alpha'))   - c_i(a_i^\star(\alpha)).
\end{align*}
Summing up the above two equations, we arrive at
\begin{align*}
  (\alpha'-\alpha)\cdot   \left(\mathbb{E}_{X\sim p_i(a_i^\star(\alpha'))}[ v(X)] -\mathbb{E}_{X\sim p_i(a_i^\star(\alpha))}[v(X)] \right)\geq  0.
\end{align*}
Since $\alpha'>\alpha$, we know that the above equation is equivalent to 
\begin{align*}
  \mathbb{E}_{X\sim p_i(a_i^\star(\alpha'))}[ v(X)] \geq\mathbb{E}_{X\sim p_i(a_i^\star(\alpha))}[v(X)],
\end{align*} This finishes the proof.
\end{proof}

With Lemma~\ref{lem:linear_property}, we can show an important property about the utility function. Precisely, we have the following lemma.
\begin{lemma}\label{lem:linear_lipschitz}
Then for any $\epsilon, \alpha \geq 0 $ with $\alpha+\epsilon\leq 1$, we have
\begin{align*}
    u(\alpha) - 
    u(\alpha+\epsilon) \leq \epsilon.
\end{align*}
\end{lemma}

\begin{proof}
From Lemma~\ref{lem:linear_property}, we know that
\begin{align*}
    u(\alpha) - 
    u(\alpha+\epsilon) \leq u(\alpha) - \frac{1-\alpha -\epsilon}{1-\alpha} u(\alpha) = \frac{\epsilon \cdot u(\alpha)}{1-\alpha}.
\end{align*}
Furthermore, from   $u(\alpha) = \mathbb{E}_{i\sim\rho}\left[\mathbb{E}_{X\sim p_i(a^\star_i(\alpha))}[(1-\alpha) v(X)]\right] $ and $\mathbb{E}_{i\sim\rho}\left[\mathbb{E}_{X\sim p_i(a^\star_i(\alpha))}[ v(X)]\right]\leq 1 $, we know that $u(\alpha)\leq 1-\alpha$. Thus one has
\begin{align*}
    u(\alpha) - 
    u(\alpha+\epsilon)\leq \epsilon.
\end{align*}
\end{proof}
Now we are ready to prove the main theorem. Let $\hat \alpha^\star = \argmax_{\alpha\in\mathcal{S}} u(\alpha)$, where $\mathcal{S}$ is the discretized search space defined in Algorithm~\ref{alg:linear}.  We have
\begin{align*}
 \sum_{t=1}^T \mathbb{E}[u(\alpha^\star) - u(\alpha^{(t)})] & \leq  \sum_{t=1}^T \mathbb{E}[u(\alpha^\star) -  u(\hat \alpha^{\star})]  + \sum_{t=1}^T\mathbb{E}[u(\hat \alpha^{\star}) -  u(\alpha^{(t)})] \\
 &
 \leq T\cdot (u(\alpha^\star) -  u(\hat \alpha^{\star})) + \sqrt{\frac{T\log(T)}{\epsilon}}.
\end{align*}
It suffices to bound the discretization error $u(\alpha^\star) -  u(\hat \alpha^{\star})$. Let $\alpha'=\argmin_{\alpha'\geq \alpha^\star, \alpha'\in\mathcal{S}}(\alpha'-\alpha)$. Then we have $\alpha' \leq \alpha +\epsilon$. Thus we have
\begin{align*}
    u(\alpha^\star) -  u(\hat \alpha^{\star}) \leq u(\alpha^\star) -  u( \alpha') \leq \epsilon.
\end{align*}
Here the first inequality is from  $\hat \alpha^\star = \argmax_{\alpha\in\mathcal{S}} u(\alpha)$ and that $\alpha'\in\mathcal{S}$, the second inequality is from Lemma~\ref{lem:linear_lipschitz}. Finally, we know that 
\begin{align*}
 \sum_{t=1}^T \mathbb{E}[u(\alpha^\star) - u(\alpha^{(t)})] \leq T\epsilon + \sqrt{\frac{T\log(T)}{\epsilon}}.
\end{align*}
Taking $\epsilon = (T/\log(T))^{-1/3}$ gives the final result. 
\end{proof}

\section{Proof of Theorem~\ref{thm:linear_lower}}\label{app:proof_lower_linear}

\begin{proof}[Proof  of Theorem~\ref{thm:linear_lower}]
The proof follows along the same lines as that of Theorem~\ref{thm:general_lower}. We first divide the space $[0, 1]$ uniformly to  $\lfloor 1/\epsilon\rfloor$ disjoint lines. We design $\lfloor 1/2\epsilon\rfloor$ actions such that each action is taken when the queried contract is in one of the lines. This approximately reduces the continuum-armed bandit problem to a discrete-armed bandit problem.  Then for each instance, we choose one of the action and lower its cost so that it will be chosen in a slightly larger regime, and it will lead to a larger utility function. 

Formally,
for all the problem instances, we consider agents of the same type and fix $\vec v = (0, 1)$, which is always $1$ for the only non-null outcome.   Let each action be indexed by a vector of integer $a(k)$ with  $k\in\{0, 1,\cdots, \left\lfloor\frac{1}{2\epsilon}\right\rfloor-1\}$, whose production probability and cost are defined as follows
\begin{align*}
    p(k) &= \left(1-\frac{1}{2(1-k\epsilon)}, \frac{1}{2(1-k\epsilon)}\right),\\
    c(k) &= \frac{1}{2}  \left(\frac{k\epsilon}{1-k\epsilon} - \sum_{j=0}^{k-1} \frac{\epsilon}{1-j\epsilon}\right).
\end{align*}
We assume that $\epsilon<0.1$ throughout the proof.
To choose action $a(k)$ as the best response, the following must hold true
\begin{align*}
    p(k)\cdot (\alpha \vec v) - c(k)\geq p(k')\cdot (\alpha \vec v) - c(k'),\forall k, k'\in\{0, 1,\cdots, \left\lfloor\frac{1}{2\epsilon}\right\rfloor-1\},  k'\neq k.
\end{align*}
We can see that this is equivalent to:
\begin{align*}
    k\epsilon\leq \alpha < (k+1)\epsilon, \text{ when } k < \left\lfloor\frac{1}{2\epsilon}\right\rfloor-1; \text{ or } k\epsilon\leq \alpha \text{ when } k = \left\lfloor\frac{1}{2\epsilon}\right\rfloor-1.
\end{align*}
The maximum payoff achieved is $1/2$ when $\alpha = k\epsilon$. 

Now consider  $\lfloor 1/(8\epsilon)\rfloor$ instances, indexed by $l\in[2, 4, 6,\cdots, 2\times \lfloor 1/(8\epsilon) \rfloor]$. For each of the instances, we keep the same probability distribution for each arm as $p(k)$, but change the cost to the following:
\begin{align*}
    c^{(l)}(k) = \begin{cases}  
     c(k) - \frac{\epsilon^2}{10}, & k=l \\ 
    c(k) - \frac{\epsilon^2}{20}, & k = 2\times \lfloor \frac{1}{8\epsilon} \rfloor+2 \\
    c(k), & \text{otherwise} \\ 
    \end{cases}
\end{align*}
Then we can see that for the instance $l$, the contracts that action $a(l)$ is chosen as the best response are those
inside the following set: \begin{align*}
\mathcal{C}_{l} = \left\{\alpha\mid   l\epsilon - \frac{(1-l\epsilon)(1-(l-1)\epsilon)\epsilon}{5}\leq \alpha<  (l+1)\epsilon + \frac{(1-l\epsilon)(1-(l+1)\epsilon)\epsilon}{5}\right\}.
\end{align*}
The maximum payoff achieved is  $1/2+ (1-(l-1)\epsilon)\epsilon/10\geq 1/2 + {\epsilon}/{20}$ when $\alpha = l\epsilon - \frac{(1-l\epsilon)(1-(l-1)\epsilon)\epsilon}{5}$.

Now consider another instance index by $2\times \lfloor \frac{1}{8\epsilon} \rfloor + 2$. We keep the same probability distribution for each arm as $p(k)$, but change the cost to the following:
\begin{align*}
    c^{(2\times \lfloor \frac{1}{8\epsilon)} \rfloor + 2}(k) = \begin{cases}  
    c(k) - \frac{\epsilon^2}{20}, & k = 2\times \lfloor \frac{1}{8\epsilon} \rfloor + 2, \\
    c(k), & \text{otherwise} \\ 
    \end{cases}
\end{align*}
For this instance, the maximum payoff is $1/2 + (1-(2\times \lfloor \frac{1}{8\epsilon} \rfloor +1)\epsilon)\epsilon/20\in [1/2+{\epsilon}/{30},  1/2+{\epsilon}/{25}]$ when  $\alpha = (2\times \lfloor \frac{1}{8\epsilon} \rfloor +2)\epsilon - \frac{(1-(2\times \lfloor \frac{1}{8\epsilon} \rfloor +2)\epsilon)(1-(2\times \lfloor \frac{1}{8\epsilon} \rfloor +1)\epsilon)\epsilon}{10}, \forall i\geq 1$. Moreover, the contracts that action $a(2\times \lfloor \frac{1}{8\epsilon} \rfloor + 2)$ is chosen as the best response are those
inside the following set: 
\begin{align*}
\mathcal{C}_{2\times \lfloor \frac{1}{8\epsilon} \rfloor + 2} =& \Big\{\alpha \mid 
   (2\times \lfloor \frac{1}{8\epsilon} \rfloor +2)\epsilon - \frac{(1-(2\times \lfloor \frac{1}{8\epsilon} \rfloor +2)\epsilon)(1-(2\times \lfloor \frac{1}{8\epsilon} \rfloor +1)\epsilon)\epsilon}{20}\leq \alpha, \\ 
   &\quad  \alpha <  (2\times \lfloor \frac{1}{8\epsilon} \rfloor +3)\epsilon + \frac{(1-(2\times \lfloor \frac{1}{8\epsilon} \rfloor +3)\epsilon)(1-(2\times \lfloor \frac{1}{8\epsilon} \rfloor +2)\epsilon)\epsilon}{20} \Big\}.
\end{align*}

Let $\mathcal{I}'$ be the set of instances $l$ with $1\leq l\leq 2\times \lfloor \frac{1}{8\epsilon} \rfloor$, $\mathcal{I}''$  be the set of $\mathcal{I}'$ unioned with  an extra instance $2\times \lfloor \frac{1}{8\epsilon} \rfloor + 2$.  
Let $u^{l}$ denote the payoff for the principal under instance $l$. 
For each of the instances, we have the following properties:
\begin{itemize}
\item The sets $\mathcal{C}_{l}$ are disjoint for different pairs of $l\in\mathcal{I}''$.
    \item For any instance $l\in\mathcal{I}''$, the action $a(l)$ is the only action that leads to optimal payoff for the principal. The optimal contract for instance $l$ lies in the set $\mathcal{C}_{l}$.
    \item  For any $l\in\mathcal{I}''$ and $\alpha\not \in\mathcal{C}_{l}$, the suboptimality $u^{l}(\alpha^\star) - u^{l}(\alpha)$ is at least $\epsilon/100.$
\end{itemize}

Now we apply Le Cam's Lemma to show the lower bound. 
From the reduction of  estimation to testing (see e.g. Prop 15.1 of~\cite{wainwright2019high}),
we know that
\begin{align}
    \sup_{i\in \mathcal{I}} \sum_{t=1}^T\mathbb{E}[u_i(\alpha^\star)- u_i(\alpha^{(t)})]\geq  &  \sup_{l\in\mathcal{I}''} \sum_{t=1}^T\mathbb{E}[u^{l}(\alpha^\star)- u^{l}(\alpha^{(t)})] \nonumber  \\ 
    \geq & \frac{1}{|\mathcal{I}''|}\sum_{l\in\mathcal{I}''} \sum_{t=1}^T\mathbb{E}[u^{l}(\alpha^\star)- u^{l}(\alpha^{(t)})]  \nonumber \\ 
    \geq & \frac{\epsilon}{100}\cdot \inf_{\Psi}\sum_{t=1}^T \frac{1}{|\mathcal{I}''|} \sum_{l\in\mathcal{I}''} \mathbb{P}_t^{l}(\Psi\not =l),\label{eq.testing_lower_linear}
\end{align}
where the infimum is taken over all measurable tests $\Psi$ that maps the history observation to the indices $\mathcal{I}''$. $\mathbb{P}_t^{l}$ is the distribution of the observations up to time $t$ under the instance $l$. 
Now from the tree based lower bound (see Lemma 3 in~\cite{gao2019batched}),
we have
\begin{align}
 \inf_{\Psi} \frac{1}{|\mathcal{I}''|} \sum_{l\in\mathcal{I}'} \mathbb{P}_t^{l}(\Psi\not =l)  & \geq   \frac{1}{|\mathcal{I}''|}\cdot \sum_{l\in\mathcal{I}'} (1-\|\mathbb{P}_t^{2\times \lfloor \frac{1}{8\epsilon} \rfloor + 2} -\mathbb{P}_t^{l} \|_{\mathsf{TV}})\nonumber  \\ 
     & \geq \frac{1}{|\mathcal{I}''|}\cdot \sum_{l\in\mathcal{I}'} \exp(-D_{\mathsf{KL}}(\mathbb{P}_t^{2\times \lfloor \frac{1}{8\epsilon} \rfloor + 2}, \mathbb{P}_t^{l} )).\label{eq.tree_lower_linear}
     \end{align}
Now we further upper bound the KL divergence
$D_{\mathsf{KL}}(\mathbb{P}_t^{2\times \lfloor \frac{1}{8\epsilon} \rfloor + 2}, \mathbb{P}_t^{l} )$. Let $T_t(l)$ be the number of pulls of the regime $\mathcal{C}_{l}$ anterior to the current batch of $t$. Note that the distributions of outcome in the two instances are always the same  when we query the contract that lies outside  $ \mathcal{C}_{l}$.  Moreover, we have $\sum_{l} T_t(l)\leq t$ since  the regimes $\mathcal{C}_{l}$ do not overlap with each other and the total time step is $t$.  By the divergence decomposition lemma~\citep{lattimore2020bandit}, we have
    \begin{align}
   \lefteqn{ \frac{1}{|\mathcal{I}''|}\cdot \sum_{l} \exp(-D_{\mathsf{KL}}(\mathbb{P}_t^{2\times \lfloor \frac{1}{8\epsilon} \rfloor + 2}, \mathbb{P}_t^{l} ))} \nonumber \\ 
    & \stackrel{(i)}{\geq}  \frac{1}{|\mathcal{I}''|}\cdot \sum_{l} \exp(-\sup_{j\in\{-1,0,1\}}D_{\mathsf{KL}}(p(l+j), p(l)) \mathbb{E}[T(l)]) \nonumber \\  
    & \stackrel{(ii)}{\geq}  \frac{1}{|\mathcal{I}''|}\cdot \sum_{l} \exp(-\sup_{j\in\{-1,0,1\}} D_{\mathsf{\mathcal{X}^2}}(p(l+j), p(l)) \mathbb{E}[T(l)]) \nonumber \\ 
    & \stackrel{(iii)}{\geq}   \frac{1}{|\mathcal{I}''|}\cdot \sum_{l} \exp(-C\cdot \epsilon^2 \mathbb{E}[T(l)])  \nonumber \\ 
     & \stackrel{(iv)}{\geq}  \exp\left(-\frac{C}{|\mathcal{I}''|}\cdot \sum_{l}\epsilon^2 \mathbb{E}[T(l)]\right) \nonumber  \\ 
      & \stackrel{(v)}{\geq} \exp(-C T\cdot 8 \epsilon^3).\label{eq.KL_lower_linear}
\end{align}
Here (i) arises from the fact that the observation distributions are the same under both instances when the queried contract is not in $\mathcal{C}(l)$; when the queried contract is in $\mathcal{C}_{l}$, the observed distribution  under the instance $2\times \lfloor \frac{1}{8\epsilon} \rfloor + 2$  is one of the  neighborhood distributions in $p(l+j)$, where $j\in\{-1,0,1\}$, while the distribution remains $p(l)$ under the instance $l$. (ii) is from that  $D_{\mathsf{KL}}$ is upper bounded by $D_{\mathcal{X}^2}$. (iii) is from the plug-in of two distributions into the $\mathcal{X}^2$-divergence. (iv) is from Jensen's inequality. (v) is from the fact that $\sum_{l} \mathbb{E}[T_t(l)]\leq t\leq T$.
Combining Equation (\ref{eq.testing_lower_linear}), (\ref{eq.tree_lower_linear}) and (\ref{eq.KL_lower_linear}) gives that
\begin{align*}
      \sup_{i\in \mathcal{I}} \sum_{t=1}^T\mathbb{E}[u_i(\alpha^\star)- u_i(\alpha^{(t)})] \geq \frac{\epsilon}{100} \cdot T \cdot \exp(-C T\cdot 8\epsilon^3).
\end{align*}
By taking $\epsilon = \frac{T^{-1/3}}{8}$ we have that
\begin{align*}
     \inf_{\pi} \sup_{\mathcal{I}} \sum_{t=1}^T\mathbb{E}[u(\alpha^\star)- u(\alpha^{(t)})]& \geq {C\cdot T^{2/3}}.
     \end{align*}
\end{proof}

\end{document}